\documentclass[12pt, draftclsnofoot, onecolumn]{IEEEtran}

\usepackage{graphicx}
\ifCLASSINFOpdf
\else
\fi
\usepackage{amsmath}
\usepackage{amsfonts}
\usepackage{amssymb}
\usepackage{multicol}
\usepackage{lipsum}
\usepackage{float}
\usepackage{cite}

\newtheorem{theorem}{Theorem}
\newtheorem{corollary}{Corollary}

\newtheorem{remark}{Remark}
\newenvironment{proof}{{\noindent \it Proof:}}{\hfill $\blacksquare$\par}
\usepackage{array}
\begin{document}
%
\title{Low-Complexity Frequency Domain Equalization over Fast Fading Channels}
%
%
%

\author{Hongyang Zhang, 
		Xiaojing Huang,~\IEEEmembership{Senior Member,~IEEE},\\ and J. Andrew Zhang,~\IEEEmembership{Senior Member,~IEEE}   
\thanks{H. Zhang, X. Huang and J. Andrew Zhang are with the School of Electrical and Data Engineering, and the Global Big Data Technologies Centre (GBDTC),  University of Technology Sydney, Ultimo, NSW, 2007, Australia (emails: Hongyang.Zhang-1@student.uts.edu.au, Xiaojing.Huang@uts.edu.au, and Andrew.Zhang@uts.edu.au).}
}

%
%

\markboth{}
{Submitted paper}
%



\maketitle

\begin{abstract}
Wireless communications over fast fading channels are challenging, requiring either frequent channel tracking or complicated signaling schemes such as orthogonal time frequency space (OTFS) modulation. In this paper, we propose low-complexity frequency domain equalizations to combat fast fading, based on novel discrete delay-time and frequency-Doppler channel models. Exploiting the circular stripe diagonal nature of the frequency-Doppler channel matrix, we introduce low-complexity frequency domain minimum mean square error (MMSE) equalization for OTFS systems with fully resolvable Doppler spreads. We also demonstrate that the proposed MMSE equalization is applicable to conventional orthogonal frequency division multiplexing (OFDM) and single carrier frequency domain equalization (SC-FDE) systems with short signal frames and partially resolvable Doppler spreads.  After generalizing the input-output data symbol relationship, we analyze the equalization performance via channel matrix eigenvalue decomposition and derive a closed-form expression for the theoretical bit-error-rate.  Simulation results for OTFS, OFDM, and SC-FDE modulations verify that the proposed low-complexity frequency domain equalization methods can effectively exploit the time diversity over fast fading channels. Even with partially resolvable Doppler spread, the conventional SC-FDE can achieve performance close to OTFS, especially in fast fading channels with a dominating line-of-sight path.    
\end{abstract}


\IEEEpeerreviewmaketitle

%
\IEEEpeerreviewmaketitle

\section{Introduction}
%
%
%
%
\IEEEPARstart {A}{chieving} reliable and efficient wireless communications in high mobility scenarios is challenging, mainly due to the difficulty in channel estimation and tracking, and the complexity in channel equalization \cite{8515088}. In current fourth generation (4G) mobile communication systems, orthogonal frequency division multiplexing (OFDM) and single-carrier frequency domain equalization (SC-FDE) are the key techniques. 
OFDM converts a frequency-selective time domain channel to multiple orthogonal frequency domain channels, and enables simple equalization and flexible resource allocation. 
SC-FDE is a kind of precoded OFDM with improved frequency diversity and power efficiency. However, there are also some disadvantages associated with OFDM and SC-FDE, such as sensitivity to the carrier frequency offset (CFO)~\cite{5282370}. When the CFO is caused by the difference between the transmitter's and receiver's local oscillators and/or the Doppler effect is resulted solely from the relative motion between the transmitter and receiver, the time-varying wireless channel can be converted into a time-invariant one after compensation for the CFO and/or the Doppler scaling~\cite{7055904}. However, when the relative motion between transmitter and receiver is not the only source of channel variation, the signal recovery needs more complicated estimation and equalization strategy. For example, the CFOs in multiple signal paths of an autonomous vehicle system may be caused by other moving reflectors such as cars and trains. In such scenarios, the wireless channel is hard to equalize, making Doppler effect a serious handicap in the emerging fifth generation (5G) systems for high mobility applications.

\par In conventional wireless communication systems, short frames of transmitted signals are used such that the channel fading can be assumed constant within a frame and the variation of the channel with time can be tracked from frame to frame. As such, the time-varying channel is treated as slow fading when the channel coherence time is large relative to the frame length. However, the channel variation cannot be treated as constant in fast fading channels where the channel coherence time is much smaller than the frame length. Many iterative algorithms such as maximum-a-posteriori (MAP) have been utilized in channel estimation and equalization to improve the performance of OFDM and SC-FDE in fast fading channels \cite{5524045,1561198,5961652}. In addition, channel tracking via Kalman filter or other techniques can reduce the pilot overhead significantly \cite{7582545,4527200,5946232}. However, besides the high complexity of such iterative algorithms, the lack of the capability to fully explore time diversity also limits the performance of the conventional wireless transmission schemes in fast fading channels.

\par The recently proposed orthogonal time frequency space (OTFS) modulation shows outstanding performance in fast fading channels~\cite{7925924}~\cite{8058662}, with advantages of both high spectral efficiency and relatively low peak-to-average power ratio (PAPR)~\cite{8599041}. OTFS is formulated in a two dimensional (2D) data plane and its signals can be represented in both delay-Doppler domain and frequency-time domain~\cite{8580850}. Using a long signal frame, this modulation technique can overcome the difficulty in equalizing signals with multiple Doppler frequency shifts and multipath fading, and hence can exploit the diversity in both time and frequency domains.  Current research on OTFS has been mostly focused on developing more efficient equalization techniques to fully exploit channel diversity. In~\cite{8503182}, an equalizer based on Markov chain Monte Carlo (MCMC) and a channel estimation method using the pseudo-random noise (PN) pilots are proposed. In~\cite{8424569} and~\cite{8377159}, the signal structure of OTFS in matrix form is analyzed and an equalizer based on the message passing (MP) algorithm is proposed. A simple sparse input-output relation for OTFS is derived in~\cite{8516353}, and  OTFS is presented in a more general form, called asymmetric OFDM, in~\cite{8599041}. OTFS shows the same performance as OFDM in static channels, but has the capability of adapting to fast fading channels. However, the above mentioned equalization techniques have some drawbacks. For example, they apply iterative algorithms to recover the signal, which greatly increases the computational complexity. The classical minimum mean square error (MMSE) equalization technique can hardly be directly applied to OTFS since the equalization requires matrix inversion which is prohibitively complicated due to the large size of the channel matrix. 

\par The high complexity in the OTFS equalization is linked to its currently widely used signal and channel models. In most of the recent studies on OTFS, e.g., in \cite{8859227} \cite{8686339}, the effective channel matrix is constructed as a product of permutation matrices and diagonal matrices to incorporate the signal multipath delays and Doppler shifts. It does not provide any explicit or systematic sparsity structure to be exploited in either time or frequency domain for complexity reduction. 
In addition, the multipath delays and Doppler shifts in such signal models are quantized to the delay-Doppler grid, which limits the modeling accuracy. 
Furthermore, most models are only developed for OTFS modulation based on the two-dimensional data matrix, e.g., in \cite{8424569} \cite{8756831}, and thus cannot be applied to other conventional modulations. More general signal and channel models are needed to characterize the fast fading channels and hence more suitable modulations can be proposed to fully exploit the channel diversity.  

\par 
In this paper, we first revisit the various channel representations of the fast fading channels and derive some novel signal and channel models in both time and frequency domains. We then demonstrate the circular stripe diagonal nature of the frequency-Doppler channel matrix and propose a low-complexity frequency domain MMSE equalization algorithm to combat fast fading channels. We also derive the discrete received signal models for conventional systems where the Doppler spread incurred in the fast fading channel is only partially resolvable. Using the new discrete signal models we show that MMSE equalization can also be applied in systems with conventional modulation, such as OFDM and SC-FDE, to improve their performance over fast fading channels. Finally, we analyze the theoretical performance of MMSE equalization and establish the relationships among the output SNR, the fast fading channel, and the signal modulation through eigenvalue decomposition of the channel matrix. Simulation results using OTFS, OFDM, and SC-FDE modulations under both line-of-sight (LOS) and non-line-of-sight (NLOS)  fast fading channels validate the theoretical analysis and show that the proposed MMSE equalization works well for both long and short signal frames, achieving large time diversity. Especially, SC-FDE is shown to achieve performance  close to that of OTFS, exemplifying how conventional modulations can benefit from the proposed low-complexity MMSE equalization techniques in fast fading channels. 
\par The contributions of this paper are summarized as follows.

\begin{itemize}
	\item First, the derived frequency domain signal model with the circular stripe diagonal frequency-Doppler channel matrix is valid for arbitrary multipath delays and Doppler shifts, enabling low-complexity frequency domain equalization to combat fast channel fading and paving the way for new modulation techniques such as OTFS to be adopted in practical systems for high mobility applications. 
	
	\item Second, the 
	derived signal models and equalization methods for short signal frame enable an existing standard-compatible system to improve its performance in fast fading channels without requiring any change in signaling protocol.
	
	\item Third, the proposed theoretical performance analysis method via channel matrix eigenvalue decomposition provides a new tool for better analyzing and understanding the impact of channel and signal modulation on  system performance.
	
	\item Finally, it is also demonstrated through simulation that our novel low-complexity frequency domain equalization method can work with imperfect channel estimation, which proves the suitability of the proposed technique for practical applications.
	
\end{itemize}

\par The rest of the paper is organized as follow. In Section \uppercase\expandafter{\romannumeral2}, the relationships among different channel representations are revisited. In Section \uppercase\expandafter{\romannumeral3}, both time and frequency domain received signal models over a fast fading channel are derived and a concise channel matrix expression is derived. In Section \uppercase\expandafter{\romannumeral4}, the MMSE equalization methods are proposed in both time and frequency domains under both fully and partially resolvable Doppler spread conditions. The complexity of the frequency domain MMSE equalization is also evaluated. In Section \uppercase\expandafter{\romannumeral5}, the theoretical equalization performance is analyzed, and then in Section \uppercase\expandafter{\romannumeral6} the simulation results are presented. Finally, conclusions are drawn in Section \uppercase\expandafter{\romannumeral7}. 

\section{Existing Channel and Signal Models}

\par In this section, we first revisit various existing representations of the fast fading channel and then present the conventional signal models in the time domain. 
\par Considering a single-input single-output (SISO) system and assuming that the continuous signal waveform transmitted over the channel is $s(t)$, the received signal can be expressed as 
\begin{align}\label{channelmodel}
r\left ( t \right )=\int_{-\infty}^{+\infty}\int_{-\infty}^{+\infty} h\left ( \tau , \nu  \right )s\left ( t-\tau  \right )e^{\mathbf{j}2\pi \nu t}d\tau d\nu +w\left ( t \right ),
\end{align}
where $h(\tau,\nu)$ is called the \textit{delay-Doppler spreading function}, $\mathbf{j}=\sqrt{-1}$ and $w(t)$ is the additive white Gaussian noise. For a sparse $P$-path channel, $h(\tau,\nu)$ is defined as
\begin{align}\label{hmodel}	
h(\tau,\nu)=\sum_{i=1}^{P}h_{i}\delta(\tau-\tau_{i})\delta(\nu-\nu_{i}),
\end{align}
where $h_{i}$, $\tau_{i}$, and $\nu_{i}$ are the path gain, delay and Doppler shift of the $i$-th path, respectively, and $\delta(\cdot)$ denotes the Dirac delta function satisfying
\begin{align}\label{eq1}
\int _{-\infty}^{+\infty}\delta(x)dx=1.
\end{align}

\par In addition to the \textit{delay-Doppler} representation, the fast fading channel can also be expressed in different domains. Applying the inverse Fourier transform (IFT) to $h(\tau,\nu)$ with respect to the Doppler frequency $\nu$, we obtain the \textit{delay-time} representation as 
\begin{align}\label{delaytimerepresentation}	
\begin{split}
h_{t}(\tau,t)= \int^{+\infty}_{-\infty} h(\tau,\nu)e^{\mathbf{j}2\pi \nu t}d\nu.
\end{split}
\end{align}
Similarly, applying Fourier transform (FT) to $h(\tau,\nu)$ with respect to the delay $\tau$, we obtain the \textit{frequency-Doppler} representation as 
\begin{align}\label{eq1}	
H_{\nu}(f,\nu)= \int^{+\infty}_{-\infty} h(\tau,\nu)e^{-\mathbf{j}2\pi f\tau}d\tau.
\end{align}
Applying both IFT and FT to $h(\tau,\nu)$ with respect to $\nu$ and $\tau$ respectively, we obtain the \textit{frequency-time} representation, also called time-frequency (TF) transfer function, as 
\begin{align}\label{TFtransfer}	
H(f,t)= \int^{+\infty}_{-\infty}\int^{+\infty}_{-\infty} h(\tau,\nu)e^{\mathbf{j}2\pi \nu t}e^{-\mathbf{j}2\pi f\tau}d\tau d\nu.
\end{align}
The relationship between $h_{t}(\tau,t)$ and $H_{\nu}(f,\nu)$ can be expressed as
\begin{align}\label{TDFDrelationship}	
H_{\nu}(f,\nu) = \int^{+\infty}_{-\infty}\int^{+\infty}_{-\infty} h_{t}(\tau,t)e^{-\mathbf{j}2\pi f\tau}e^{-\mathbf{j}2\pi \nu t} d\tau dt.
\end{align}
We see that $H_{\nu}(f,\nu)$ is the two-dimensional (2D) FT of $h_{t}(\tau,t)$. Fig. 1 shows the relationships among the various fast fading channel representations, which characterize the same fast fading channel in different ways.
\begin{figure}[t]
	\centering
	\includegraphics[width=1\linewidth]{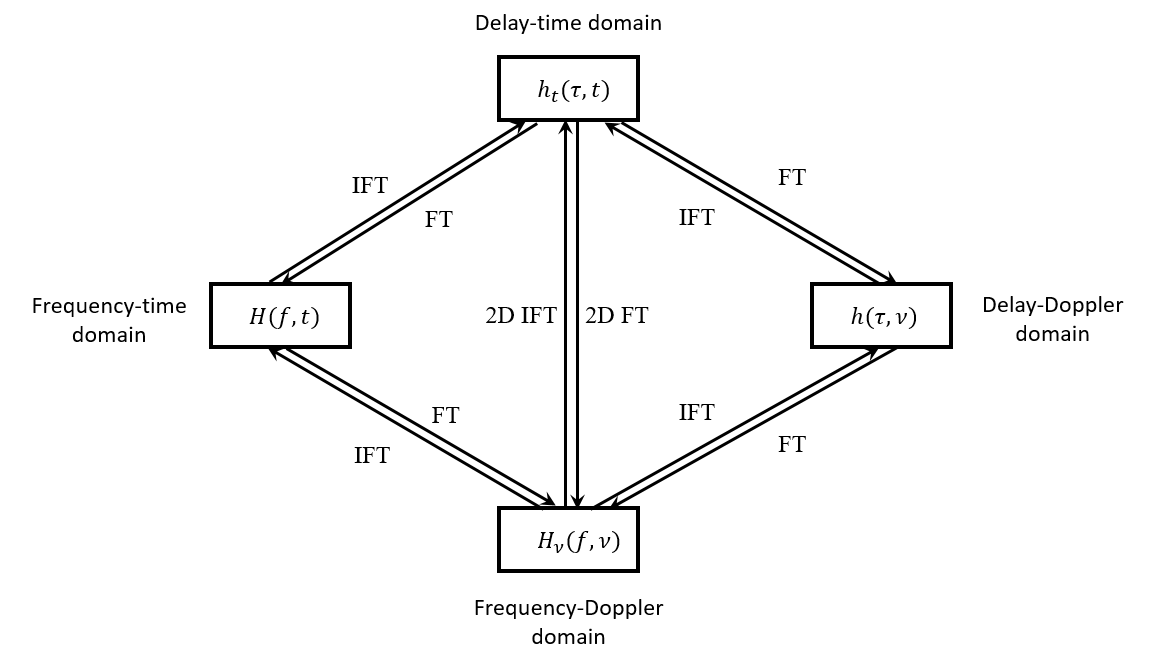}
	\caption{Relationships among different channel representations.}
	\label{fig:activeIndex}
\end{figure}
\par Note that signal models adopted in most of existing studies on OTFS systems, e.g., \cite{8424569,8859227,8686339} are analyzed in the time domain with following form
\begin{align}\label{formermodel}
r\left ( t \right )&=\int_{-\infty}^{+\infty}\int_{-\infty}^{+\infty} h\left ( \tau , \nu  \right )s\left ( t-\tau  \right )e^{\mathbf{j}2\pi \nu (t-\tau)}d\tau d\nu +w\left ( t \right )\nonumber\\
&=\int_{-\infty}^{+\infty}\int_{-\infty}^{+\infty} h'\left ( \tau , \nu  \right )s\left ( t-\tau  \right )e^{\mathbf{j}2\pi \nu t}d\tau d\nu +w\left ( t \right ),
\end{align}
where $h'(\tau,\nu)=h(\tau,\nu)e^{-\mathbf{j}2\pi \nu \tau}$.  
In the discrete form, the time-domain received signal is expressed as 
\begin{align}\label{formermodel}
r[n]=\displaystyle\sum_{i=1}^{P}h_{i}e^{\mathbf{j}2\pi \frac{k_{i}(n-l_{i})}{MN}}s[n-l_{i}]+w[n],
\end{align}
where $s[n]$ is the discrete transmitted signal, $l_{i}$ and $k_{i}$ are on-grid delay and Doppler shift for the $i$-th path, $w[n]$ is the noise sample, and $MN$ is the number of total data symbols expressed in an $M\times N$ matrix for a transmission frame.

We should notice that with this old model, the relationships illustrated in Fig.1 are no longer valid since the phase term $e^{-\mathbf{j}2\pi \nu \tau}$ is imposed on the delay-Doppler spreading function. As a result, the previous studies on OTFS seldom consider frequency-domain equalization.

\section{Novel Signal and Channel  Models for Fast Fading Channels} 
\par With the delay-time and frequency-Doppler representations of fast fading channels, we can derive novel received signal models in both time domain and frequency domain. In particular, our novel frequency domain models enable the design of low complexity equalization. In this section, model analysis in each aspect will all start from a continuous form and then transform into a discrete form to illustrate special constructions of channel matrices.
 
\subsection{Models in Time Domain}
Firstly, substituting \eqref{delaytimerepresentation} into \eqref{channelmodel} 
and letting $\tau ' = t-\tau$, we have
\begin{align}\label{TDreceive}	
r(t)&=\int^{+\infty}_{-\infty} h_{t}(t-\tau',t)s(\tau')d\tau'+w(t), 
\end{align}
which is the time domain received signal model over the considered fast fading channel. We see that the time domain received signal is a convolution of the time domain transmitted signal with the time-varying impulse response $h_{t}(\tau,t)$ in terms of delay $\tau$. 
\par We assume that the discrete-time transmitted signal $s[i]$ is a sequence of sampled continuous-time signal $s(id_{r})$ where $d_{r}$ denotes the sampling period which is also termed as the \textit{delay resolution}, and is of finite length modulated from a set of $M\times N$ origin data symbols $x[i], i=0,1,...,MN-1$.
\par In the matrix form, the data symbol set can be expressed as an $M\times N$ matrix $\mathbf{X}=[\mathbf{x}_{0},\mathbf{x}_{1},...,\mathbf{x}_{N-1}]$ where
\begin{align}\label{eq1}
\mathbf{x}_{n}=(x[nM],x[nM+1],...,x[nM+M-1])^{T},\nonumber\\
n=0,1,...,N-1,
\end{align}
is a column vector and  $(\cdot)^{T}$ denotes the transpose of a matrix.
We further assume that the maximum channel multipath delay is $d_{max}$ and the maximum Doppler frequency is $f_{max}$. Hence, the maximum number of resolvable multipaths is $L_{max} = \left \lceil d_{max}/d_{r} \right \rceil $, where $\left \lceil \cdot \right \rceil$ denotes the ceiling function to obtain the rounded up number, provided that the channel has a minimum bandwidth $1/d_{r}$, and the maximum number of resolvable Doppler frequencies (positive or negative side) is $K_{max} = \left \lceil f_{max}/f_{r} \right \rceil$, where $f_{r}$ is the \textit{Doppler resolution}, provided that the transmitted signal has a minimum frame length $1/f_{r}$.

Under the above assumptions, the discrete delay-time representation of the fast fading channel can be obtained as
\begin{align}\label{DTDrepresentation}
h_{t}[i,j] = h_{t}(id_{r},jd_{r})=\int_{-1/2d_{r}}^{1/2d_{r}}H(f,jd_{r})e^{\mathbf{j}2\pi fid_{r}}df.
\end{align}
For the sparse $P$-path channel expressed in Eq. \eqref{hmodel}, its TF transfer function can be expressed as
\begin{align}\label{DTDhmodel}
H(f,t)=\sum_{i=1}^{P}h_{i} e^{-\mathbf{j}2\pi f\tau_{i}}e^{\mathbf{j}v2\pi \nu_{i}t}.
\end{align}
\par Denote the discrete-time received signal sequence as a vector $\mathbf{r}$ and the transmitted signal sequence as a vector $\mathbf{s} = [s[0],s[1],...,s[MN-1]]^{T}$. From Eq. \eqref{TDreceive}, the discrete time domain received signal model can be expressed in the matrix form as  
\begin{align}\label{TDmodel}
\mathbf{r}=\mathbf{H}_{t}\mathbf{s+w},
\end{align}
where $\mathbf{w}$ denotes the noise vector and $\mathbf{H}_{t}$ is the delay-time channel matrix defined as
\begin{align}\label{Htmatrix}
\mathbf{H}_{t}=\begin{bmatrix}
h_{t}[0,0] &\cdots     &h_{t}[1,0] \\ 
h_{t}[1,1]& \cdots   & h_{t}[2,1]\\ 
\vdots  & \ddots    &\vdots  \\ 
h_{t}[MN-1,MN-1] & \cdots    & h_{t}[0,MN-1])
\end{bmatrix}.
\end{align}
\par Note that due to the signal discretization in both time and frequency domains, the linear convolution in \eqref{TDreceive} becomes a circular convolution and thus the delay-time domain channel matrix is constructed from a periodically extended $h_{t}[i,j]$ as illustrated in Fig. 2, where the shaded squares with different gray levels indicate different discrete values of $h_{t}[i,j]$.   We see that the coordinate transformation $\tau = t-\tau'$ maps the parallelogram enclosed by the dashed lines in $\tau-t$ coordinates onto a squared area in the $t-\tau'$ coordinates. It is also interesting to see that when the channel is time-invariant, $\mathbf{H}_{t}$ becomes a circulant matrix composed of the channel's impulse response.

\par Also note that the realization of circular convolution requires some signaling overhead in the transmitted signal frame, i.e., a cyclic prefix (CP) or zero-padded suffix with length longer than the maximum multipath delay should be inserted before or after the signal frame. This will be further illustrated in Section IV.B.
\begin{figure}[t]
	\centering
	\includegraphics[width=1\linewidth]{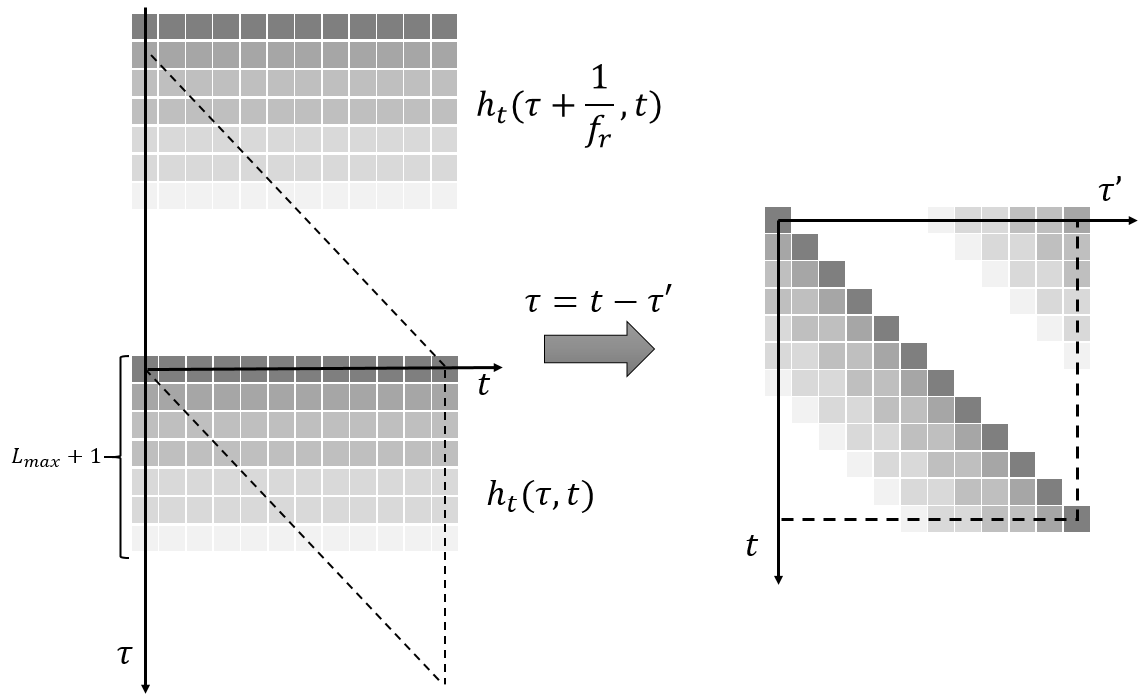}
	\caption{Construction of delay-time domain channel matrix.}
	\label{fig:activeIndex}
\end{figure}
\begin{remark}
	Different to conventional models, the delay-time domain channel matrix $\mathbf{H}_{t}$ introduced in this paper can be constructed for any multipath delay $\tau_{i}$ and Doppler shift $\nu_{i}$ through the TF transfer function as shown in Eqs. \eqref{DTDrepresentation} and \eqref{DTDhmodel}. Referring to the conventional discrete-time domain signal model used in most of the recent works \cite{8424569, 8686339, 8515088} based on Eq. \eqref{formermodel}, an effective channel matrix can be constructed as a product of many permutation matrices (representing the circular shifts of input signal samples) and phase-only diagonal matrices (representing the Doppler shifts) \cite{8516353}. Though the effective channel matrix demonstrates some sparsity, it is channel dependent and lacks the explicit time-varying convolution nature. In addition, such effective channel matrix is only valid for on-grid multipath delays and Doppler shifts. 
\end{remark}

\subsection{Frequency Domain Signal Models}

\par In this section, we first derive the frequency domain channel model which receives less attention in existing OTFS work and then present the concise structure of its channel matrix. 

Representing the transmitted signal in the frequency domain as $S(f)$ and substituting
\begin{align}\label{eq1}	
s(t)=\int^{+\infty}_{-\infty}S(f)e^{\mathbf{j}2\pi ft}df,
\end{align}
into \eqref{channelmodel}, we have
\begin{align}\label{eq1}	
r(t)&=\int^{+\infty}_{-\infty}\int^{+\infty}_{-\infty}\int^{+\infty}_{-\infty} h\left ( \tau , \nu  \right )S\left ( f  \right )e^{\mathbf{j}2\pi f(t-\tau)}e^{\mathbf{j}2\pi \nu t}d\tau d\nu df\nonumber\\
&+w(t)\nonumber\\
&=\int^{+\infty}_{-\infty}\int^{+\infty}_{-\infty} h_{t}\left ( \tau , t  \right )S\left ( f  \right )e^{\mathbf{j}2\pi f(t-\tau)}d\tau df+w(t)\nonumber\\
&=\int^{+\infty}_{-\infty} H(f,t)S(f)e^{\mathbf{j}2\pi ft}df+w(t).
\end{align}
Applying FT to $r(t)$, the frequency domain received signal can be modeled as
\begin{align}\label{FDreceive}	
R(f)&=\int^{+\infty}_{-\infty} r(t)e^{-\mathbf{j}2\pi ft}dt+W(f)\nonumber\\
&=\int^{+\infty}_{-\infty}\int^{+\infty}_{-\infty} H(f',t)e^{-\mathbf{j}2\pi (f-f')t}dt S(f')df'+W(f)\nonumber\\
&=\int^{+\infty}_{-\infty} H_{\nu}(f',f-f')S(f')df'+W(f),
\end{align}
where $W(f)$ is the additive white Gaussian noise in the frequency domain. The channel frequency response  at Doppler frequency $\nu$, $H_{\nu}(f,\nu)$, can also be  interpreted as a frequency-dependent Doppler response. We see that the frequency domain received signal is a convolution of the frequency domain transmitted signal with the frequency-dependent Doppler response in terms of Doppler frequency $\nu$.

\par Then, the discrete frequency-Doppler representation of the fast fading channel can be obtained as
\begin{align}\label{FDhmatrix}
H_{\nu}[i,j] = H_{\nu}(if_{r},jf_{r})=\int_{0}^{1/f_{r}}H(if_{r},t)e^{-\mathbf{j}2\pi f_{r}t}dt.
\end{align}
\par Denote the discrete received frequency signal as a vector R and the transmitted frequency domain signal as a vector $\mathbf{S} = [S[0],S[1],...,S[MN-1]]^{T}$ where $S[i]$ is the discrete Fourier transform (DFT) of $s[i]$. From Eq. \eqref{FDreceive}, the discrete frequency domain received signal model can be expressed in the matrix form as

\begin{align}\label{FDmodel}
\mathbf{R=H_{\nu}S+W},
\end{align}
where $\mathbf{W}$ denotes the frequency domain noise vector and $\mathbf{H}_{\nu}$ is the frequency-Doppler channel matrix defined as
	\begin{align}\label{Hvmatrix}
	\mathbf{H_{\nu}}=\begin{bmatrix}
	H_{\nu}[0,0]  & \cdots & H_{\nu}[MN-1,1] \\ 
	H_{\nu}[0,1]  & \cdots & H_{\nu}[MN-1,2]\\ 
	\vdots   & \ddots &\vdots  \\ 
	H_{\nu}[0,MN-1]  & \cdots  & H_{\nu}[MN-1,0]
	\end{bmatrix}.
	\end{align}
\par The frequency-Doppler channel matrix $\mathbf{H_{\nu}}$ is also constructed from a periodically extended $H_{\nu}[i,j]$ as illustrated in Fig. 3 where the shaded squares with different gray levels indicate different discrete values of $H_{\nu}[i,j]$. We also see that the coordinate transformation $\nu = f-f'$ maps the parallelogram enclosed by the dashed lines in $\nu-f'$ coordinates onto a squared area in $f-f'$ coordinates. 
\begin{figure}[t]
	\centering
	\includegraphics[width=1\linewidth]{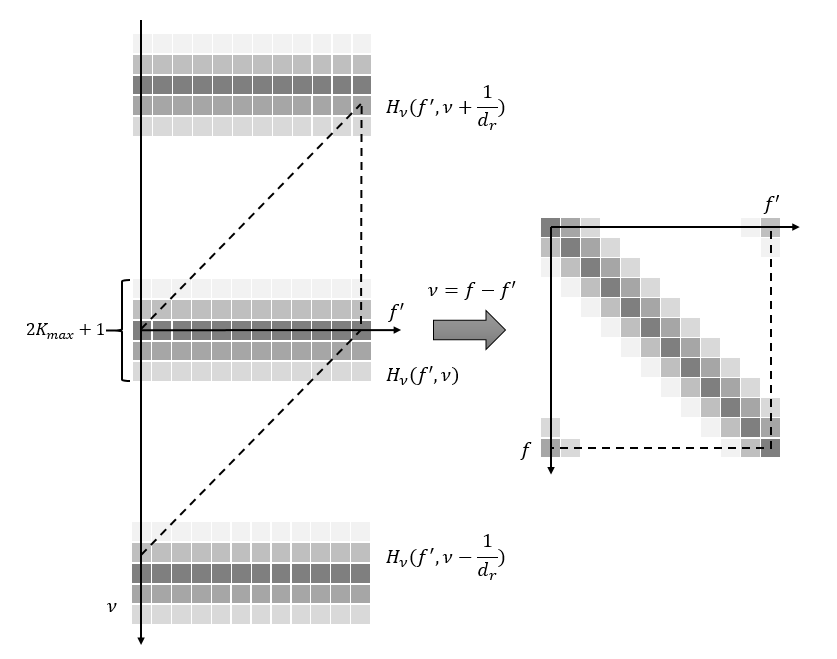}
	\caption{Construction of frequency-Doppler domain channel matrix.}
	\label{fig:activeIndex}
\end{figure}
\begin{remark}
	Similar discrete-frequency signal models demonstrating the circular stripe diagonal property in their channel matrices can be found in the literature \cite{BOOK2} for OFDM systems. However, such signal models do not consider the resolvability of the channel Doppler spread, and are also only valid for on-grid multipath delays and Doppler shifts. Due to the relatively short frame length of the OFDM symbols, the actual Doppler shifts incurred in the fast fading channel may not be fully resolvable. On the contrary, our proposed frequency-Doppler domain channel matrix $\mathbf{H}_{\nu}$ can be constructed for any multipath delay $\tau_{i}$ and Doppler shift $\nu_{i}$ through the TF transfer function as shown in Eqs. \eqref{FDhmatrix} and \eqref{DTDhmodel}.
\end{remark}

\subsection{Connections between Time-Domain and Frequency-Domain Models}
\par The time-varying impulse response and frequency-dependent Doppler response of the fast fading channel can be used to fully describe the input-output relationships of time domain and frequency domain signals, respectively. In addition to Eq. \eqref{TDFDrelationship}, the two channel representations are related to each other as described in the following theorem.

\begin{theorem}
	Given the delay-time and frequency-Doppler representations of a fast fading channel, $h_{t}(\tau,t)$ and $H_{\nu}(f,\nu)$, its time-varying impulse response can be expressed as $h_{t}(t-\tau',t)$, where the delay is defined as $\tau = t-\tau'$ referenced at time $\tau'$. Meanwhile, its frequency-dependent Doppler response can be expressed as $H_{\nu}(f',f-f')$, where the Doppler frequency is defined as $\nu=f-f'$ referenced at frequency $f'$. $h_{t}(t-\tau',t)$ and $H_{\nu}(f',f-f')$ have the relationship
	\begin{align}\label{TDFDtrans}	
	&H_{\nu}(f',f-f') \nonumber\\
	&= \int^{+\infty}_{-\infty}\int^{+\infty}_{-\infty} h_{t}(t-\tau',t)e^{\mathbf{j}2\pi f'\tau'}e^{-\mathbf{j}2\pi ft} d\tau' dt.
	\end{align}
\end{theorem}

\begin{proof}
	Let $\nu=f'-f$ and express Eq. \eqref{TDFDrelationship} as
	\begin{align}\label{TDFDtrans2}
	&H_{\nu}(f,f'-f) \nonumber\\
	&=\int_{-\infty}^{+\infty}\int_{-\infty}^{+\infty}h_{t}(\tau,t)e^{-\mathbf{j}2\pi f\tau}e^{-\mathbf{j}2\pi (f'-f) t} d\tau dt\nonumber\\
	&=\int_{-\infty}^{+\infty}\int_{-\infty}^{+\infty}h_{t}(\tau,t)e^{\mathbf{j}2\pi f(t-\tau)}e^{-\mathbf{j}2\pi f't} d\tau dt.
	\end{align}
	Substituting $t-\tau$ with a new variable $\tau'$,  we have $\tau=t-\tau'$ and hence $d\tau = -d\tau'$ at any given $t$. Then, Eq. \eqref{TDFDtrans2} is further expressed as 
	\begin{align}\label{eq1}
	&H_{\nu}(f,f'-f) \nonumber\\
	&= -\int_{-\infty}^{+\infty}\int_{+\infty}^{-\infty} h_{t}(t-\tau',t)e^{\mathbf{j}2\pi f\tau'}e^{-\mathbf{j}2\pi f't} d\tau' dt\nonumber\\
	&= \int_{-\infty}^{+\infty}\int_{-\infty}^{+\infty} h_{t}(t-\tau',t)e^{\mathbf{j}2\pi f\tau'}e^{-\mathbf{j}2\pi f't} d\tau' dt.
	\end{align}
	Finally, interchanging the variables $f$ and $f'$, we have \eqref{TDFDtrans}.
\end{proof}

\par  Theorem 1 suggests that $H_{\nu}(f',f-f')$ is the IFT and FT of $h_t(t-\tau',t)$ in terms of $\tau'$ and $t$, respectively. In the discrete domain, we have the following corollary which can be easily derived from Theorem 1.
\begin{corollary}
	The relationship between the frequency-Doppler channel matrix $\mathbf{H_{\nu}}$ and the delay-time channel matrix $\mathbf{H}_{t}$  can be expressed as
	\begin{align}\label{eq1}	
	\mathbf{H}_{\nu} =\mathbf{F H}_{t}  \mathbf{F^{H}},
	\end{align}
	where $(\cdot)^{\mathbf{H}}$ denotes the conjugate and transpose of a matrix, $\mathbf{F}$ denotes the DFT matrix and  $\mathbf{F^{H}}$ denotes  the IDFT matrix. $\mathbf{F}$ and $\mathbf{F^{H}}$ satisfy $\mathbf{FF^{H}}=\mathbf{F^{H}F}=\mathbf{I}$, where $\mathbf{I}$ is the identity matrix of order $MN$. 
\end{corollary}
\par This corollary reveals a generalized relationship between time domain and frequency domain channel models. For the conventional time-invariant channel, $\mathbf{H}_{t}$ is a circulant matrix composed of the channel's impulse response, and $\mathbf{H_{\nu}}$ is a diagonal matrix composed of $\mathbf{H}_{t}$'s eigenvalues which represent the channel's frequency response. However, for time-varying channels, $\mathbf{H_{\nu}}$ is not a diagonal matrix any more but a circular stripe diagonal matrix with the stripe width equaling to $2K_{max}+1$ as shown in Fig. 3.
\section{Proposed Low-Complexity MMSE Equalization}
\par We introduce MMSE equalization based on the derived discrete received signal models in the time and frequency domains. Both systems with fully and partially resolvable Doppler spreads are considered.  
\subsection{MMSE Equalization with Fully Resolvable Doppler Spread}
As discussed in Section II, given the multipath and Doppler frequency resolutions $d_{r}$ and $f_{r}$, the transmission system requires a minimum bandwidth of $1/d_{r}$ and the transmitted signal requires a minimum length of $1/f_{r}$. Under these conditions, the received signal can be modeled in the discrete time and frequency domains as shown in \eqref{TDmodel} and \eqref{FDmodel} respectively. Further assume that the channel matrices are known. We can then recover the transmitted signal in either the time or frequency domain under the MMSE criterion. 
In the discrete time domain, an estimate of the transmitted signal after MMSE equalization can be expressed as
\begin{align}\label{eq1}
\mathbf{\hat{s}=G}_{t}\mathbf{r},
\end{align}
where $\mathbf{G}_{t}$ can be derived, following a well established process, as
\begin{align}\label{MMSETDGmatrix}
\mathbf{G}_{t}=\mathbf{H}^{\mathbf{H}}_{t}(\mathbf{H}_{t}\mathbf{H}_{t}^{\mathbf{H}}+\frac{1}{\gamma_{in}}\mathbf{I})^{-1},
\end{align}
$\gamma_{in}$ denotes the input signal-to-noise ratio (SNR) at the receiver~\cite{4392126}. 

In the discrete frequency domain, an estimate  of the transmitted signal after MMSE equalization can be expressed as
\begin{align}\label{eq1}
\mathbf{\hat{S}=G_{\nu}R},
\end{align}
where $\mathbf{G_{\nu}}$ is similarly derived as 
\begin{align}\label{MMSEFDGmatrix}
\mathbf{G_{\nu}=H_{\nu}^{H}}(\mathbf{H_{\nu}H_{\nu}^{H}}+\frac{1}{\gamma_{in}}\mathbf{I})^{-1}.
\end{align}
\par The two methods can achieve the same MMSE as stated in the following theorem.

\begin{theorem}
	Under the same channel condition, time domain and frequency domain MMSE equalizers produce the same mean square error.
\end{theorem}

\begin{proof}
The mean square error after MMSE equalization can be expressed as $tr(\mathbf{I-G}_{t}\mathbf{h}_{t})=tr(\mathbf{I})-tr(\mathbf{G}_{t}\mathbf{h}_{t})$ in the time domain and  $tr(\mathbf{I-G_{\nu}H_{\nu}})=tr(\mathbf{I})-tr(\mathbf{G_{\nu}H_{\nu}})$ in the frequency domain respectively. 
\par According to Eq. \eqref{MMSEFDGmatrix} and Corollary 1 , we have
\begin{align}\label{eq1}
&tr(\mathbf{G_{\nu}H_{\nu}})\nonumber\\
&=tr(\mathbf{H^{H}_{\nu}}(\mathbf{H_{\nu}H^{H}_{\nu}}+\frac{1}{\gamma_{in}}\mathbf{I})^{-1}\mathbf{{H}_{\nu}})\nonumber\\
&=tr(\mathbf{FH}^{\mathbf{H}}_{t}\mathbf{F^{H}}(\mathbf{FH}_{t}\mathbf{H}^{\mathbf{H}}_{t}\mathbf{F^{H}}+\frac{1}{\gamma_{in}}\mathbf{FF^{H}})^{-1}\mathbf{F{H}}_{t}\mathbf{F^{H}})\nonumber\\
&=tr(\mathbf{H}^{\mathbf{H}}_{t}\mathbf{F^{H}}(\mathbf{F}(\mathbf{H}_{t}\mathbf{H}^{\mathbf{H}}_{t}+\frac{1}{\gamma_{in}}\mathbf{I})\mathbf{F^{H}})^{-1}\mathbf{FH}_{t})\nonumber\\
&=tr(\mathbf{H^{H}}_{t}(\mathbf{H}_{t}\mathbf{H}^{\mathbf{H}}_{t}+\frac{1}{\gamma_{in}}\mathbf{I})^{-1}\mathbf{H}_{t})\nonumber\\
&=tr(\mathbf{G}_{t}\mathbf{H}_{t}),
\end{align}
which proves that the two mean square errors are the same.
\end{proof}

\par The above theorem implies that both the time domain and frequency domain MMSE equalizers achieve the same performance. However, due to the different constructions of the delay-time channel matrix and the frequency-Doppler channel matrix, the computational complexity involved in the calculation of the equalization matrix is different. Generally speaking, the number of Doppler frequencies is much less than the number of multipaths to be resolved in the transmission system. Therefore, performing MMSE equalization in the frequency domain has a significant advantage in terms of receiver complexity. This is similar to the case that conventional frequency domain one-tap equalization is less complicated than the time domain linear equalization over a time-invariant channel. In general, we have the following remark regarding the MMSE equalization complexity.
\begin{remark}
	The computational complexity of the frequency domain MMSE equalization is $O((1+4K_{max})^{2}MN)$  in terms of the number of complex multiplications and divisions, where $MN$ is the length of the signal frame and $K_{max}$ is the number of positive (or negative) Doppler frequencies. The details are shown in Appendix A.
\end{remark}
\par Note that if conventional MMSE equalization is used without exploring the structure of the frequency-Doppler domain channel matrix, the computational complexity will be $O((MN)^{3})$. An example is provided here to demonstrate the advantage of frequency-domain MMSE equalization. In a practical scenario as defined in the ETSI 5G channel models (for more details see Table I in Section VI), we have $M= 256$, $N=32$ and $K_{max} =3$. With frequency domain MMSE equalization, the number of complex multiplications/divisions is in the order of $1.4\times 10^{6}$, whereas with time domain MMSE equalization it is $5.5\times 10^{11}$.

\subsection{MMSE Equalization with Partially Resolvable Doppler Spread}

\par If the signal frame is shorter than $1/f_{r}$, the received signal will not be able to resolve all the Doppler frequencies. This is the case for many existing systems. Without loss of generality, we assume that the short signal frame has $M$ samples with length $T=Md_{r}$ whereas the long signal frame has $MN$ samples with length $NT$. We also assume that the CPs have the same length $T_{cp}=L_{cp}d_{r}$ for both the long and short signal frames, where $L_{cp}$ is the number of samples in a CP. Fig. 4 provides a comparison between long and short signal frames. Given the Doppler resolution $f_{r}$, one period of the frequency-time domain representation of the fast fading channel is also illustrated to show the channel variation. We see that the short signal frame only experiences a part rather than a full period of the channel variation. 
\begin{figure}[t]
	\centering
	\includegraphics[width=1\linewidth]{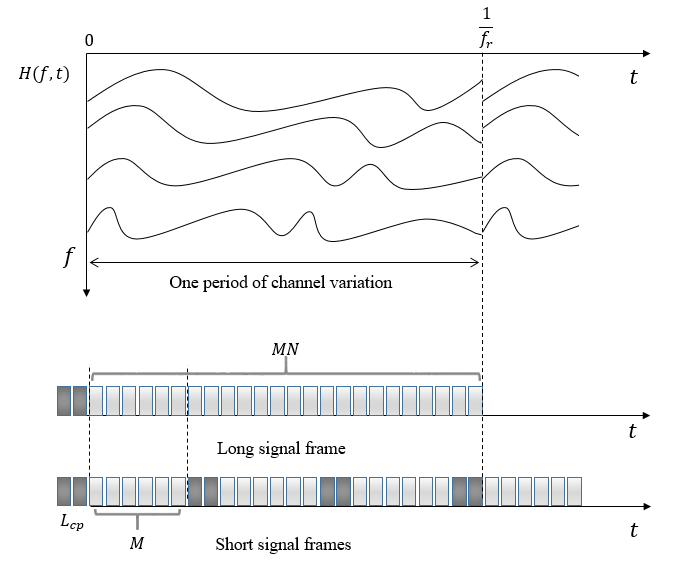}
	\caption{Comparison between long and short signal frames.}
	\label{fig:activeIndex}
\end{figure}
\par Suppose that the first short signal frame is aligned with the long signal frame. After passing through the same fast fading channel, the $n$-th received short-frame signal, $n=0,1,...N-1$, can be expressed in the time domain as
\begin{align}\label{OFDMTD}
\mathbf{r}^{(n)} = \mathbf{H}_{t}^{(n)}\mathbf{s}^{(n)} + \mathbf{w}^{(n)},
\end{align}
where $\mathbf{s}^{(n)}$ and $\mathbf{w}^{(n)}$ are the $n$-th transmitted signal vector and noise vector respectively, and  $\mathbf{H}_{t}^{(n)}$ is the $M\times M$ delay-time channel matrix which can be constructed similar to Eq. \eqref{Htmatrix} and is shown in Fig. 2 except that the size of the matrix is $M\times M$ and the delay-time channel representation for the $n$-th short signal frame is now $h_t^{(n)}(\tau,t)=h_t(\tau,t+n(T+T_{cp}))$ for $0\leq t<T$.
In the frequency domain, the $n$-th short-frame received signal can be expressed as
\begin{align}\label{OFDMFD}
\mathbf{R}^{(n)} = \mathbf{H}_{\nu}^{(n)}\mathbf{S}^{(n)} + \mathbf{W}^{(n)},
\end{align}
where $\mathbf{S}^{(n)}$ and $\mathbf{W}^{(n)}$ are the $n$-th frequency domain transmitted signal vector and noise vector respectively, and  $\mathbf{H}_{\nu}^{(n)} $is the $M\times M$ frequency-Doppler channel matrix which can be constructed similar to Eq. \eqref{Hvmatrix}, with reference to Fig. 3. However, the discrete frequency-Doppler channel representation for the $n$-th short signal frame should now be calculated as
\begin{align}\label{OFDMHv}
&H_{\nu}^{(n)}[i,j]\nonumber\\
&=\sum_{j'=-K_{max}}^{K_{max}}e^{\frac{\mathbf{j}2\pi(M+L_{cp})nj'}{MN} }H_{\nu}[iN,j'] \phi(\frac{2\pi}{MN}(jN-j')), \nonumber\\
&0\leq i,j\leq M-1,
\end{align}
where $\phi(w)$ is the Fourier transform of a discrete rectangular window function defined as 
\begin{align}\label{eq1}
\phi(\omega)=\frac{\sin(\omega M/2)}{M\sin(\omega /2)}e^{-\mathbf{j}\omega (M-1)/2}.
\end{align}
\par The calculation expressed in Eq. \eqref{OFDMHv} can be explained as follows.  In order to calculate the discrete frequency-Doppler representation of the fast fading channel during the $n$-th short signal frame, we need its frequency-time representation which is a time shifted version of $H(f,t)$ advanced by $n(T+T_{cp})$, i.e., $H^{(n)}(f,t)=H(f,t+n(T+T_{cp}))$ for $0\leq t<T$,   and then weighted by a rectangular window function of width $T$. The windowing becomes a convolution in the frequency-Doppler representation. Eq. \eqref{OFDMHv} is the discrete version of such convolution. The result is finally down-sampled by $N$ times as the frequency and Doppler resolution is reduced by $N$ times.

\par Similar to Corollary 1, we have  $\mathbf{H}_{\nu}^{(n)} = \mathbf{F}_{M} \mathbf{H}_{t}^{(n)}\mathbf{F}^{\mathbf{H}}_{M}$, where $\mathbf{F}_{M}$ and $\mathbf{F}^{\mathbf{H}}_{M}$ are the DFT and IDFT matrices, respectively,  satisfying  $\mathbf{F}_{M} \mathbf{F}^{\mathbf{H}}_{M} = \mathbf{F}^{\mathbf{H}}_{M}\mathbf{F}_{M} =\mathbf{I}_M$ and $\mathbf{I}_M$ is the identity matrix of order $M$. 

\par Given the discrete signal models in both the time and frequency domains shown in \eqref{OFDMTD} and \eqref{OFDMFD}, the MMSE equalization methods for the short signal frames can be derived accordingly similar to \eqref{MMSETDGmatrix} and \eqref{MMSEFDGmatrix}, and the same equalization performance can be achieved in either time or frequency domain. Due to the reduced frame length, the equalization complexity is also reduced.
\begin{remark}
	Though the signal model expressed in \eqref{OFDMFD} is similar to those found in the literature, e.g., \cite{BOOK2} and those mentioned in Remark 2, the frequency-Doppler channel matrix $\mathbf{H}^{(n)}_{\nu}$ has different meanings, i.e., each line in the diagonal stripe of the matrix does not represent a resolvable Doppler frequency, but a convolution of the frequency-Doppler transfer function with the Fourier transform of a windowing function which reflects the shorter signal frame length. In addition, considering a full period of channel time variation given the Doppler resolution, the location of the short signal frame also has an impact on $\mathbf{H}^{(n)}_{\nu}$, which is reflected by the phase factor in \eqref{OFDMHv}. 
\end{remark}
\section{Performance Analysis}
In this section, we show that the proposed MMSE equalization can be applied to various systems with reduced complexity and improved performance. To do so, we first formulate the equalized output data symbols in relation to the transmitted ones for various modulation schemes, then provide the SNR analysis for the output data symbols after MMSE equalization, and finally derive the theoretical  bit error rate (BER).
\subsection{Input-Output Relationships}
\par Three modulation schemes, OTFS, OFDM and SC-FDE, are considered here. Regarding Doppler resolution, OTFS can fully resolve all the Doppler frequencies due to its long signal frame, whereas OFDM and SC-FDE can only resolve part of them due to shorter signal frames. 
\par Assume that the total number of data symbols to be transmitted is $MN$ and let $\mathbf{X} \in \mathbb{C}^{M\times N}$ denote the $M\times N$ 2D OTFS data symbol matrix. After the inverse symplectic finite Fourier
transform (ISFFT), Heisenberg transform and pulse shaping, CP is prepended to the signal frame. The modulated OTFS signal is then transmitted over a fast fading channel and a time-domain sequence $r[i], i=0,1,...MN-1$, is received.  According to \cite{8516353}  and the channel model \eqref{TDmodel}, the received signal can be expressed in the matrix form
\begin{align}\label{eq1}
\mathbf{r}=\mathbf{H}_{t}(\mathbf{F}^{\mathbf{H}}_{N}\otimes \mathbf{I}_{M})\mathbf{x}+\mathbf{w},
\end{align}
where $\mathbf{F}^{\mathbf{H}}_{N}$ denotes the $N$-point IFFT matrix, $\otimes$ denotes Kronecker product, $\mathbf{x}=vec(\mathbf{X})$ is the vectorized form of matrix $\mathbf{X}$ and $\mathbf{w}$ is the noise vector. 
The received data symbols after MMSE equalization can be expressed as
\begin{align}\label{eq1}
\mathbf{y}&=(\mathbf{F}_{N}\otimes \mathbf{I}_{M})\mathbf{G}_{t}\mathbf{H}_{t}(\mathbf{F}^{\mathbf{H}}_{N}\otimes \mathbf{I}_{M})\mathbf{x}+(\mathbf{F}_{N}\otimes \mathbf{I}_{M})\mathbf{G}_{t}\mathbf{w}.
\end{align}
Similarly, letting $\mathbf{x}^{(n)}$ denote the $n$-th transmitted OFDM or SC-FDE data symbol vector, the $n$-th received OFDM data symbol vector after MMSE equalization can be expressed as
\begin{align}\label{eq1}
\mathbf{y}^{(n)}&=\mathbf{F}_{M}\mathbf{G}_{t}^{(n)}\mathbf{H}_{t}^{(n)}\mathbf{F}^{\mathbf{H}}_{M}\mathbf{x}^{(n)}+\mathbf{F}_{M}\mathbf{G}_{t}^{(n)}\mathbf{w},
\end{align}
where  $\mathbf{G}_{t}^{(n)}$ denotes the time domain equalization matrix for the $n$-th OFDM symbol.
Assuming SC-FDE has the same  frame structure, the $n$-th received data symbol vector for SC-FDE can be expressed as 
\begin{align}\label{eq1}
\mathbf{y}^{(n)}&=\mathbf{G}_{t}^{(n)}\mathbf{H}_{t}^{(n)}\mathbf{x}^{(n)}+\mathbf{G}_{t}^{(n)}\mathbf{w}.
\end{align}

\par Note that the received data symbol expressions for OTFS, OFDM and SC-FDE have some similarity. Hence, a general representation can be written as
\begin{align}\label{Generalmodel}
\mathbf{y}&=\mathbf{V^{H}}\mathbf{G}_{t}\mathbf{H}_{t}\mathbf{V}\mathbf{x}+\mathbf{V^{H}}\mathbf{G}_{t}\mathbf{w},
\end{align}
where $\mathbf{V}$ and $\mathbf{V^{H}}$ denote the signal modulation and demodulation matrix respectively, satisfying $\mathbf{VV^{H}=V^{H}V=I}_{MN}$ or $\mathbf{I}_{M}$. For example, $\mathbf{V}$ is $\mathbf{F}^{\mathbf{H}}_{N}\otimes \mathbf{I}_{M}$ for OTFS, $\mathbf{F}^{\mathbf{H}}_{M}$ for OFDM, and $\mathbf{I}_{M}$ for SC-FDE. Note that for different modulations, the  size of $\mathbf{G}_{t}$, $\mathbf{H}_{t}$,  $\mathbf{x}$ and $\mathbf{y}$ may be different. The size of both $\mathbf{G}_{t}$ and $\mathbf{H}_{t}$ is $MN\times MN$ for OTFS with $MN\times 1$ signal vectors $\mathbf{x}$ and $\mathbf{y}$, whereas the size of $\mathbf{G}_{t}$ and $\mathbf{H}_{t}$ are $M\times M$ for OFDM and SC-FDE with $M\times 1$ signal vectors $\mathbf{x}$ and $\mathbf{y}$. 
\subsection{Output SNR Analysis}
In the following analysis, we utilize the general representation with time domain MMSE equalization  in \eqref{Generalmodel} to determine the output SNR after equalization, assuming the data matrix size is $M \times N$. According to Theorem 2, the output SNR will be the same if frequency domain MMSE equalization is used. The same analysis can also be applied to OFDM and SC-FDE but with smaller channel matrix dimensions.  For simplicity, we also define $\mathbf{A = V^{H}G}_{t}\mathbf{H}_{t}\mathbf{V}$ and $\mathbf{B=V^{H}G}_{t}$.
Assuming that the data symbols are independent of each other with the average power $\sigma_{x}^{2}$, i.e., $\mathbf{E\{xx^{H}}\}=\sigma_{x}^{2}\mathbf{I}$ and the noise power is $\sigma^{2}_{w}$, i.e., $\mathbf{E\{ww^{H}}\}= \sigma^{2}_{w}\mathbf{I}$,  where $\mathbf{E}\{\cdot\}$ denotes ensemble expectation,
the covariance matrix of $\mathbf{y}$ can be derived as
\begin{align}\label{covariancey}
\mathbf{E\{yy^{H}}\}&=\mathbf{AE\{xx^{H}\}A^{H}+BE\{ww^{H}\}B^{H}}\nonumber\\
&=\mathbf{C}\sigma_{x}^{2}+\mathbf{D}\sigma_{w}^{2},
\end{align}
where $\mathbf{C}$ and $\mathbf{D}$ denote $\mathbf{AA^{H}}$ and $\mathbf{BB^{H}}$ respectively. To detect the $(nM+m)$-th data symbol at the $m$-th row and the $n$-th column in $\mathbf{X}$ , the useful signal power after equalization is
\begin{align}\label{eq1}
&|\mathbf{A}[nM+m,nM+m]|^{2}\sigma_{x}^{2}=q_{0}[m,n],
\end{align}
where $\mathbf{A}[i,j]$ denotes the element of  $\mathbf{A}$ at the $i$-th row and the $j$-th column.

The average total power of the $(nM+m)$-th element in $\mathbf{y}$ can also be expressed from \eqref{covariancey} as
\begin{align}\label{eq1}
&\mathbf{C}[nM+m,nM+m]\sigma_{x}^{2}+\mathbf{D}[nM+m,nM+m]\sigma_{w}^{2}\nonumber\\
&=q_{1}[m,n].
\end{align}
For simplicity, let $a_{m,n}$, $c_{m,n}$ and $d_{m,n}$ denote $\mathbf{A}[nM+m,nM+m]$, $\mathbf{C}[nM+m,nM+m]$ and $\mathbf{D}[nM+m,nM+m]$ respectively. Then, the output SNR after equalization can be expressed as 
\begin{align}\label{SNRbypower}
\gamma_{out}[m,n]&=\frac{q_{0}[m,n]}{q_{1}[m,n]-q_{0}[m,n]}\nonumber\\
&=\frac{|a_{m,n}|^{2}}{c_{m,n}+\frac{1}{\gamma_{in}}d_{m,n}-|a_{m,n}|^{2}}\nonumber\\
&=\frac{1}{1-\frac{|a_{m,n}|^{2}}{c_{m,n}+\frac{1}{\gamma_{in}}d_{m,n}}}-1,
\end{align}
where the input SNR is defined as $\gamma_{in} = \sigma_{x}^{2}/\sigma_{w}^{2}$.

Eq. \eqref{SNRbypower} evaluates the output SNR directly based on the equalized data symbol vector expression \eqref{Generalmodel}, where the impact of channel condition and signal modulation on the output SNR is not explicitly shown. 
To demonstrate how the output SNR is affected by the channel and the signal modulation, we express $\mathbf{A}$  as 
\begin{align}\label{eq1}
\mathbf{A}&=\mathbf{V^{H}}\mathbf{H}^{\mathbf{H}}_{t}(\mathbf{H}_{t}\mathbf{H}_{t}^{\mathbf{H}}+\frac{1}{\gamma_{in}}\mathbf{I})^{-1}\mathbf{H}_{t}\mathbf{V}\nonumber\\
&=\mathbf{V^{H}}((\mathbf{H}_{t})^{-1}\mathbf{H}_{t}\mathbf{H}_{t}^{\mathbf{H}}(\mathbf{H}_{t}^{\mathbf{H}})^{-1}+\frac{1}{\gamma_{in}}(\mathbf{H}^{\mathbf{H}}_{t}\mathbf{H}_{t})^{-1})^{-1}\mathbf{V}\nonumber\\
&=\mathbf{V^{H}}(\mathbf{I}+\frac{1}{\gamma_{in}}(\mathbf{H}^{\mathbf{H}}_{t}\mathbf{H}_{t})^{-1})^{-1}\mathbf{V}\nonumber\\
&=(\mathbf{I}+\frac{1}{\gamma_{in}}(\mathbf{V^{H}}\mathbf{H}_{t}^{\mathbf{H}}\mathbf{H}_{t}\mathbf{V})^{-1})^{-1}.
\end{align}
Here, $\mathbf{H}_{t}^{\mathbf{H}}\mathbf{H}_{t}$ is a Hermitian matrix and can be expressed through eigenvalue decomposition as $\mathbf{H}_{t}^{\mathbf{H}}\mathbf{H}_{t}=\mathbf{Q\Lambda Q^{H}}$, where $\mathbf{Q}$ is a square $MN\times MN$ unitary matrix and $\mathbf{\Lambda}$ is a diagonal matrix  with the $i$-th diagonal element $\lambda_{i}$  for $i = 0,1,...,MN-1.$ Further denoting $\mathbf{U} = \mathbf{V^{H}}\mathbf{Q}$, we can simplify $\mathbf{A}$ as
\begin{align}\label{eq1}
\mathbf{A}&=(\mathbf{I}+\frac{1}{\gamma_{in}}\mathbf{U\Lambda ^{-1}U^{H}})^{-1}\nonumber\\
&=\mathbf{U}(\mathbf{I}+\frac{1}{\gamma_{in}}\mathbf{\Lambda ^{-1}})^{-1}\mathbf{U^{H}}\nonumber\\
&=\mathbf{U}(\text{diag}(1+\frac{1}{\gamma_{in}\lambda_{i}}))^{-1}\mathbf{U^{H}}\nonumber\\
&=\mathbf{U}\text{diag}(\frac{\gamma_{in}\lambda_{i}}{\gamma_{in}\lambda_{i}+1})\mathbf{U^{H}}\nonumber\\
&=\mathbf{I}-\mathbf{U}\text{diag}(\frac{1}{\gamma_{in}\lambda_{i}+1})\mathbf{U^{H}},
\end{align}
where $\text{diag}(x_{i})$ denotes a diagonal matrix with the $i$-th diagonal element $x_{i}$. According to the MMSE equalization principle, the normalized noise power for the $(nM+m)$-th equalized data symbol can be expressed as
\begin{align}\label{noisepower}
J_{nM+m}&=1-\mathbf{A}[nM+m,nM+m]\nonumber\\
&=\sum_{i=0}^{M N-1}\frac{1}{\gamma_{in}\lambda_{i}+1}|\mathbf{U}[nM+m,i]|^{2}.
\end{align} 
Therefore, the output SNR after equalization can be expressed as 
\begin{align}\label{outputSNR}
\gamma_{out}[m,n]&=\frac{1-J_{nM+m}}{J_{nM+m}}\nonumber\\
&=\frac{1}{J_{nM+m}}-1.
\end{align}
It can be approved that \eqref{SNRbypower} and \eqref{outputSNR} are equivalent as shown in Appendix B.

\par Similarly, the SNR analysis can be performed based on the eigenvalue decomposition of $\mathbf{H_{\nu}^{H}H_{\nu}}$, which produces the same eigenvalues as those of $\mathbf{H}_{t}^{\mathbf{H}}\mathbf{H}_{t}$ but with Fourier transformed eigenvectors.


\par From \eqref{noisepower} and \eqref{outputSNR}, the output SNR is determined by the eigenvalues of $\mathbf{H}_{t}^{\mathbf{H}}\mathbf{H}_{t}$, which characterize the channel fading condition, as well as the matrix $\mathbf{U}$ that is related to the signal modulation. This relationship can help us understand why OTFS can gain both frequency and time diversity. Applying the same analysis to OFDM and SC-FDE, we can also see why OFDM cannot exploit frequency diversity but SC-FDE can. 
\subsection{Theoretical BER}
Given a specific fading channel, the average BER for an $M$-ary quadrature amplitude modulation (QAM) is
\begin{align}\label{eq1}
P_{b}&=\frac{\sum_{m=0}^{M-1}\sum_{n=0}^{N-1}\frac{2(1-2^{-k})}{k}Q(\sqrt{\frac{3}{4^{k}-1}\gamma_{out}[m,n]})}{MN},
\end{align}
where the Q-function is defined as $Q(x)=\frac{1}{\sqrt{2\pi}}\int_{x}^{\infty}e^{-\frac{t^{2}}{2}}dt$ and $2^{2k}$ indicates the modulation level. For example, $k=1$ means 4-QAM or QPSK~\cite{BOOK1}. Averaging over all possible fading channels, the ergodic BER for the fast fading channel is expressed as $E_{h}\{P_{b}\}$, where $E_{h}\{\cdot\}$ denotes the ensemble average over all delay-Doppler channel realizations.  Similarly, the  theoretical BERs of OFDM and SC-FDE can also be obtained.
\section{Simulation Results}
\par In this section, the BER performance of the frequency domain MMSE equalization is compared over fast fading channels among OTFS, OFDM and SC-FDE with uncoded 4-QAM modulation. The ETSI tapped delay line (TDL) models, which are valid for frequency range from 0.5 GHz to 100 GHz with a maximum bandwidth of 2 GHz, are adopted as the multipath channel models~\cite{ETSITDL}. TDL channel models define the time delay over LOS (more accurately, Racian channel with dominating LOS path) and NLOS channels in different scenarios. According to \cite{ETSITDL}, the simulation parameters are listed in Table I with the following assumptions. Firstly, the channel is assumed perfectly known at the receiver, i.e., ideal channel estimation is adopted. Secondly, each TDL channel model  defines several application scenarios with different time delay spreads. However, as has been reported in \cite{8855898}, very similar performance is achieved in different scenarios. As such, we only take the urban macrocell (UMa) channels as examples in our simulation. Both LOS and NLOS channels are simulated using TDL-D and TDL-A models respectively. The maximum multipath delays $d_{max}$ are $4.55$ $\mu s$ and $3.51$ $\mu s$ for LOS and NLOS channels respectively. Thirdly, the Doppler frequency shifts are added to all paths, which vary in every channel realization and obey uniform distribution over $[-K_{max},K_{max}]$. 

\begin{table}[t]
	\centering
	\caption{Simulation Parameters}
	\label{table:nonlin}
	\renewcommand\arraystretch{1.4}
	\begin{tabular}{|c|c|c|}
		\hline
		\begin{tabular}[c]{@{}c@{}}Carrier \\ Frequency\\ ($f_{c}$)\end{tabular}  & \begin{tabular}[c]{@{}c@{}}No. of \\ Subcarriers\\ ($M$)\end{tabular}         & \begin{tabular}[c]{@{}c@{}}No. of  \\ OFDM/SC-FDE\\ Symbols ($N$)\end{tabular}   \\ \hline 
		6 GHz & 256  & 32   \\ \hline
		\begin{tabular}[c]{@{}c@{}}Subcarrier\\ Spacing\\ ($\Delta f$)\end{tabular}  &  \begin{tabular}[c]{@{}c@{}}Bandwidth\\ ($W=M\Delta f$)\end{tabular}  & \begin{tabular}[c]{@{}c@{}}Duration of \\ OFDM/SC-FDE\\ Symbol\\ ($T=M/W$)\end{tabular}  \\ \hline
		30 KHz  & 7.68 MHz &  33.33 $\mu$s \\ \hline
		\begin{tabular}[c]{@{}c@{}}Delay\\ Resolution\\ ($d_{r}=1/W$)\end{tabular} & \begin{tabular}[c]{@{}c@{}}Doppler\\ Resolution\\ ($f_{r}=1/NT$)\end{tabular} & \begin{tabular}[c]{@{}c@{}}Maximum\\ Speed\\ ($v_{max}$)\end{tabular}             \\ \hline
		130.21 ns  & 937.5 Hz & 500 Km/h   \\ \hline
		\begin{tabular}[c]{@{}c@{}}Maximum\\ Doppler Frequency\\ ($f_{max}=f_{c}\frac{v_{max}}{v_{c}}$,\\ $v_{c}=3\times 10^{8} \  m/s$)\end{tabular}  &   \begin{tabular}[c]{@{}c@{}}No. of \\ Doppler  Shifts \\ (Positive or Negative) \\ ($K_{max}=\left \lceil \frac{f_{max}}{f_{r}} \right \rceil$)\end{tabular}  & \begin{tabular}[c]{@{}c@{}}No. of \\ Multipaths\\ ($L_{max}=\left \lceil \frac{d_{max}}{d_{r}} \right \rceil$)\end{tabular}   \\ \hline
		2777.8 Hz & $\approx 3$ & $\approx35$(LOS), $27$(NLOS)   \\ \hline                                                           
	\end{tabular}
\end{table}
\par Fig. 5 plots one realization of the frequency-time representation of the fast fading TDL-D channel, and the channel variation with time is clearly shown  in the figure. Fig. 6 shows the frequency-Doppler channel matrix $\mathbf{H_{\nu}}$ derived from $H(f,t)$ for OTFS with full Doppler resolution. We see that the non-zero elements appear in a narrow stripe along the diagonal lines and other elements are all zeros, which validates the analysis in Section III. According to the parameters shown in Table I, the width of the stripe is $2K_{max}+1=7$, and is very small compared with the $8192\times 8192$ matrix size. The high sparsity of $\mathbf{H}_{\nu}$ allows for more efficient matrix inversion operation as analyzed in Section IV. Therefore, the frequency domain MMSE equalization can be realized at a low complexity. The frequency-Doppler channel matrix for OFDM and SC-FDE can also be derived from $H(f,t)$, which is an $M\times M$ matrix, and an example is shown in Fig. 7. It also exhibits an obvious diagonal stripe. 
\begin{figure}[t]
	\centering
	\includegraphics[width=1\linewidth]{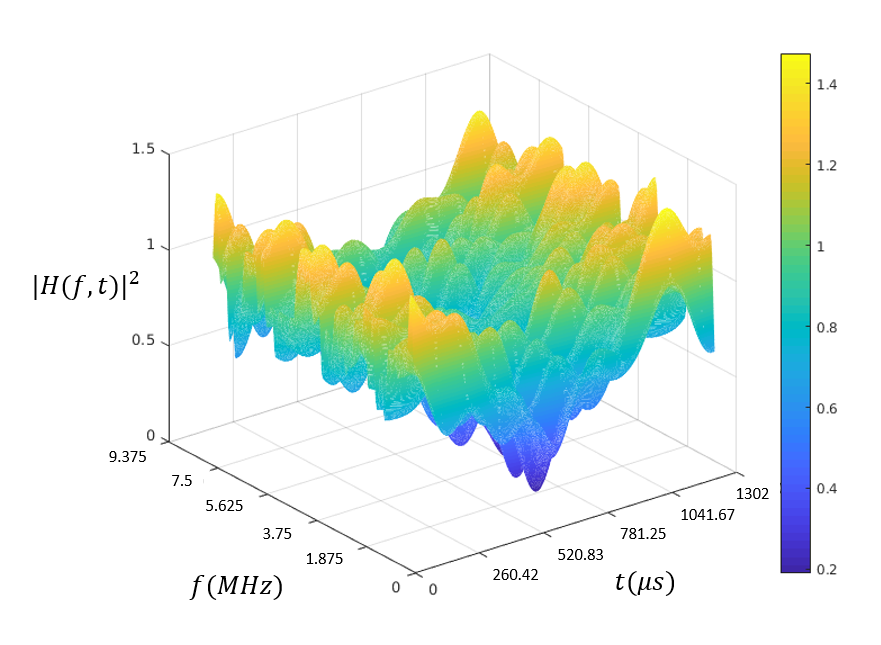}
	\caption{Frequency-time domain representation of fast fading channels.}
	\label{fig:activeIndex}
\end{figure} 
\begin{figure}[t]
	\centering
	\includegraphics[width=1\linewidth]{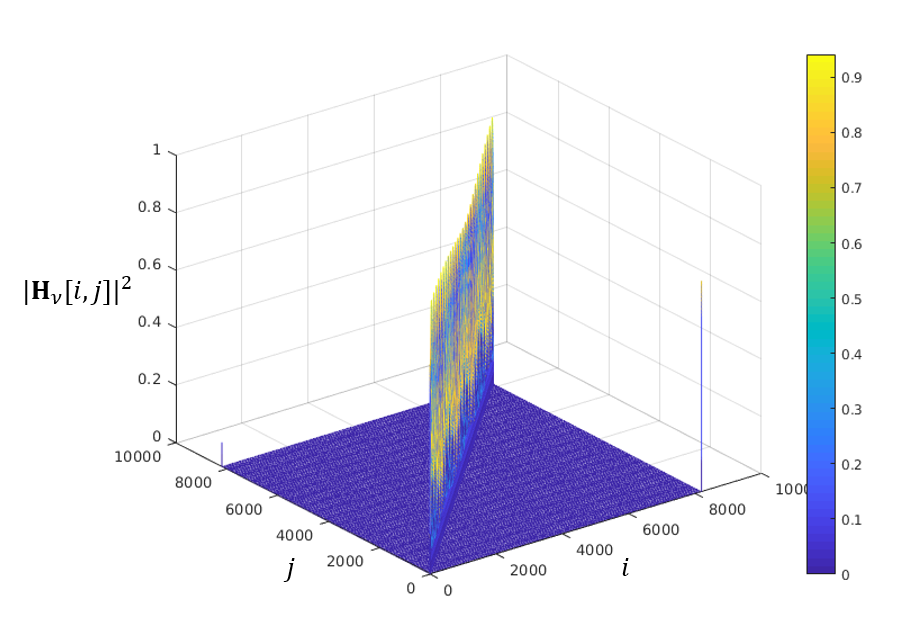}
	\caption{Frequency-Doppler domain channel matrix for OTFS over fast fading channels.}
	\label{fig:activeIndex}
\end{figure}
\begin{figure}[t]
	\centering
	\includegraphics[width=1\linewidth]{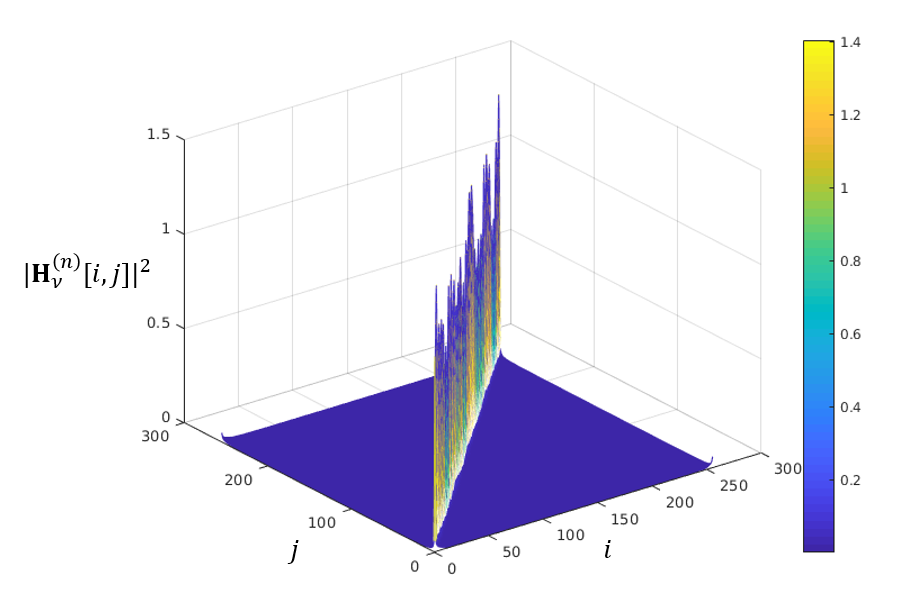}
	\caption{Frequency-Doppler domain channel matrix for OFDM and SC-FDE over fast fading channels.}
	\label{fig:activeIndex}
\end{figure}
\par The simulated BER performance of OTFS, OFDM and SC-FDE in fast fading LOS channel is shown in Fig. 8. We observe that systems with the proposed frequency domain MMSE equalization perform significantly better than those with the conventional one-tap frequency domain equalization method. The performance of OTFS is the best, with $10^{-4}$ BER at about 14 dB SNR. In comparison, at the same BER, SC-FDE has only less than 1 dB performance degradation, while OFDM has a large 4 dB degradation due to the lack of frequency diversity. We also see that the theoretical BER curves perfectly match the simulation results. Since a perfect channel state information (CSI) can not be achieved in practice, we also simulate the performance when imperfect channel estimation is applied in the channel equalization. The imperfect channel is simulated by adding a random error matrix obeying zero-mean Gaussian distribution into the estimated channel matrix, where the variance of the channel error is assumed to be inversely proportional to the SNR. The practical performance of OTFS, OFDM and SC-FDE in the LOS channel is also shown in Fig. 8. We see that the impact of channel estimation error on the performance is significant only at a lower SNR but tends to be minor at a higher SNR. 
\begin{figure}[htbp]
	\centering
	\includegraphics[width=1\linewidth]{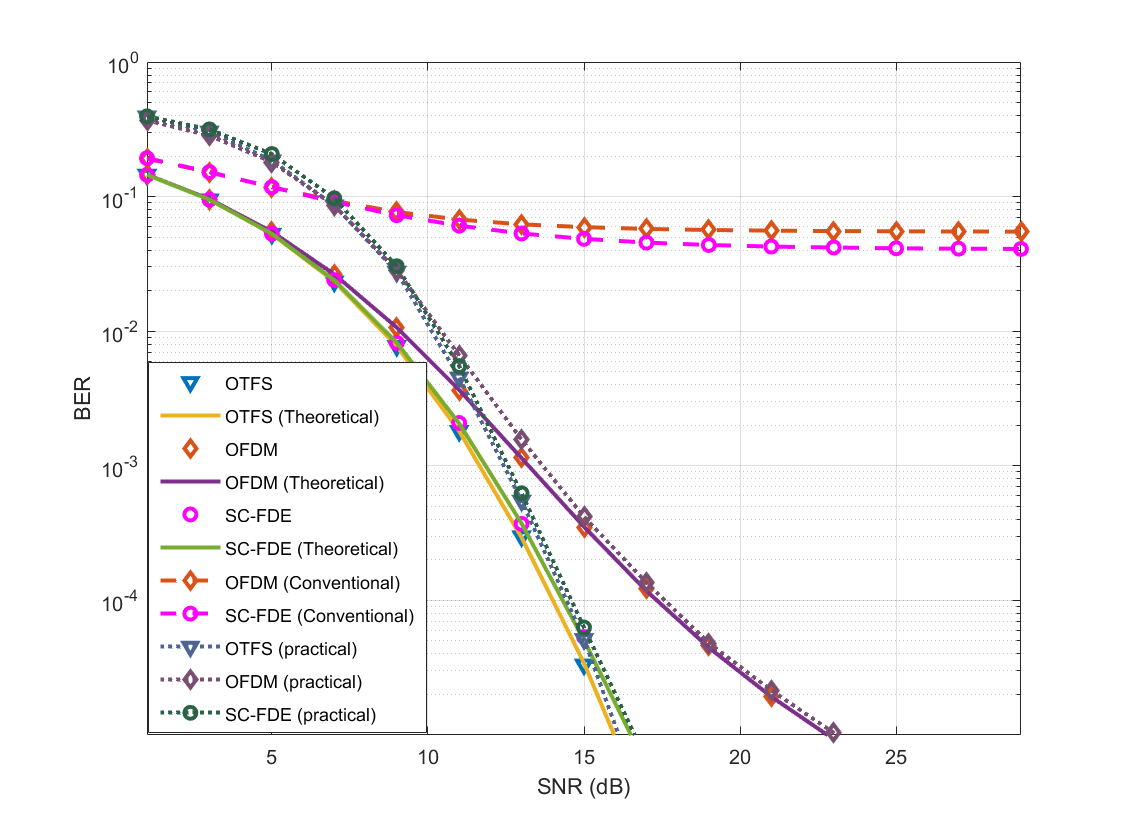}
	\caption{Performance comparison in TDL-D channel.}
	\label{fig:activeIndex}
\end{figure}
\par Fig. 9 shows the BER performance in fast fading NLOS channels.  We observe that OTFS performs much better than the other two. The performance of SC-FDE is degraded as compared with OTFS but is better than that of OFDM. Overall, the proposed MMSE equalization improves the performance of OFDM and SC-FDE in fast fading channel significantly under both LOS and NLOS conditions. SC-FDE performs very well, especially in LOS channels. The impact of imperfect channel estimation on equalization performance is also investigated in NLOS channel conditions. In Fig. 9, the practical results demonstrate a similar trend as shown in the LOS channel. These results verify our proposed equalization algorithm is feasible for coping with fast fading channel in practice. 
\begin{figure}[htbp]
	\centering
	\includegraphics[width=1\linewidth]{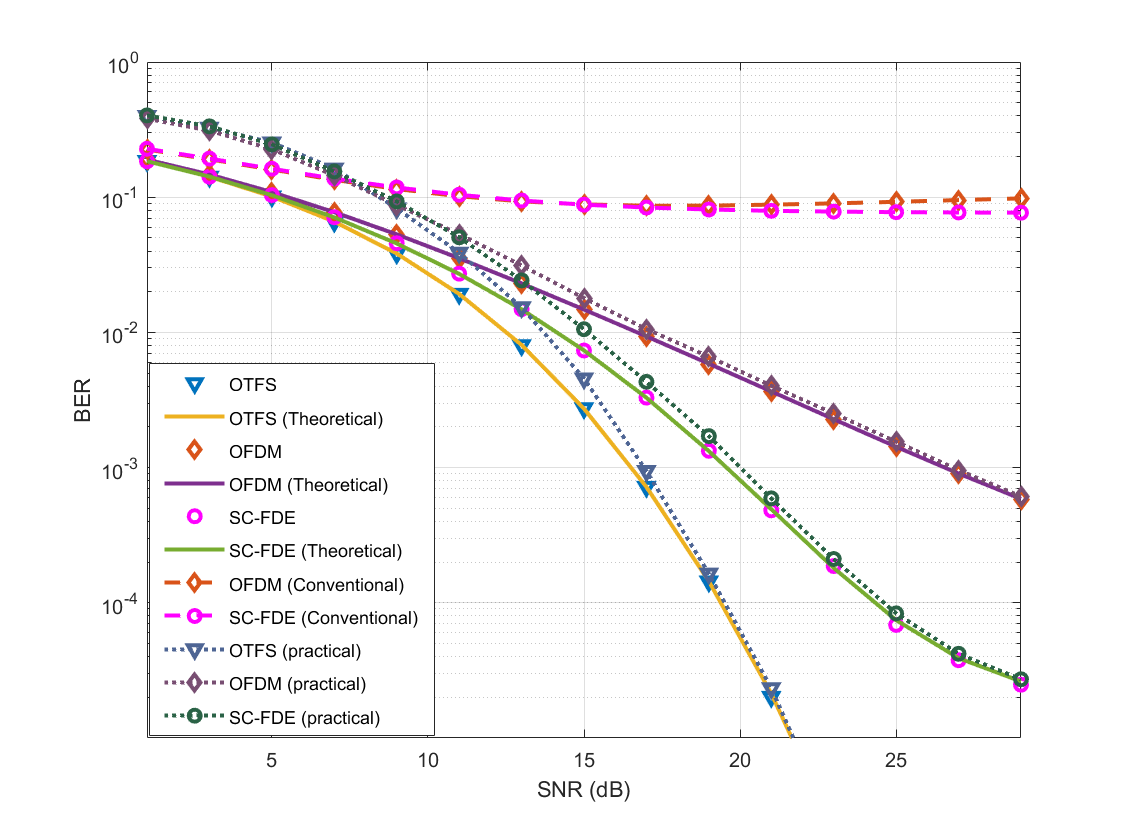}
	\caption{Performance comparison in TDL-A channel.}
	\label{fig:activeIndex}
\end{figure}

Finally, we compare the proposed low-complexity frequency domain equalization method with the popular MP algorithm proposed in \cite{8377159}. Based on the same parameter settings, the BER performance over NLOS channels is simulated and compared. Note that the MP equalization requires a number of iterations that affect both the computational complexity and the BER performance. The BER performance for MMSE and MP equalization is compared in Fig. 10, where the number of iterations is selected as 10, 20, 30, 40 and 50 for MP equalization.
The BER performance under MP is better than that of MMSE in low SNR ranges. However, MMSE outperforms MP in high SNR region. We also observe that there are error floors for MP even with a large number of iterations. To demonstrate the significant complexity reduction achieved by MMSE, we compare the calculation complexity by counting the number of arithmetic operations (multiplication, division and logarithm). Under the ETSI NLOS channels shown in Table \ref{table:nonlin} with $K_{max} = 3$ and $L_{max} = 27$, the calculation complexity for MMSE equalization is around $O(169\times MN)$, while the calculation complexity for MP is around $O(2700 \times MN\times n_{iter})$, where $n_{iter}$ is the number of iterations \cite{8424569}. It is evident that the MMSE equalization is superior to MP in term of complexity and hence is more suitable for practical applications. 
\begin{figure}[htbp]
	\centering
	\includegraphics[width=1\linewidth]{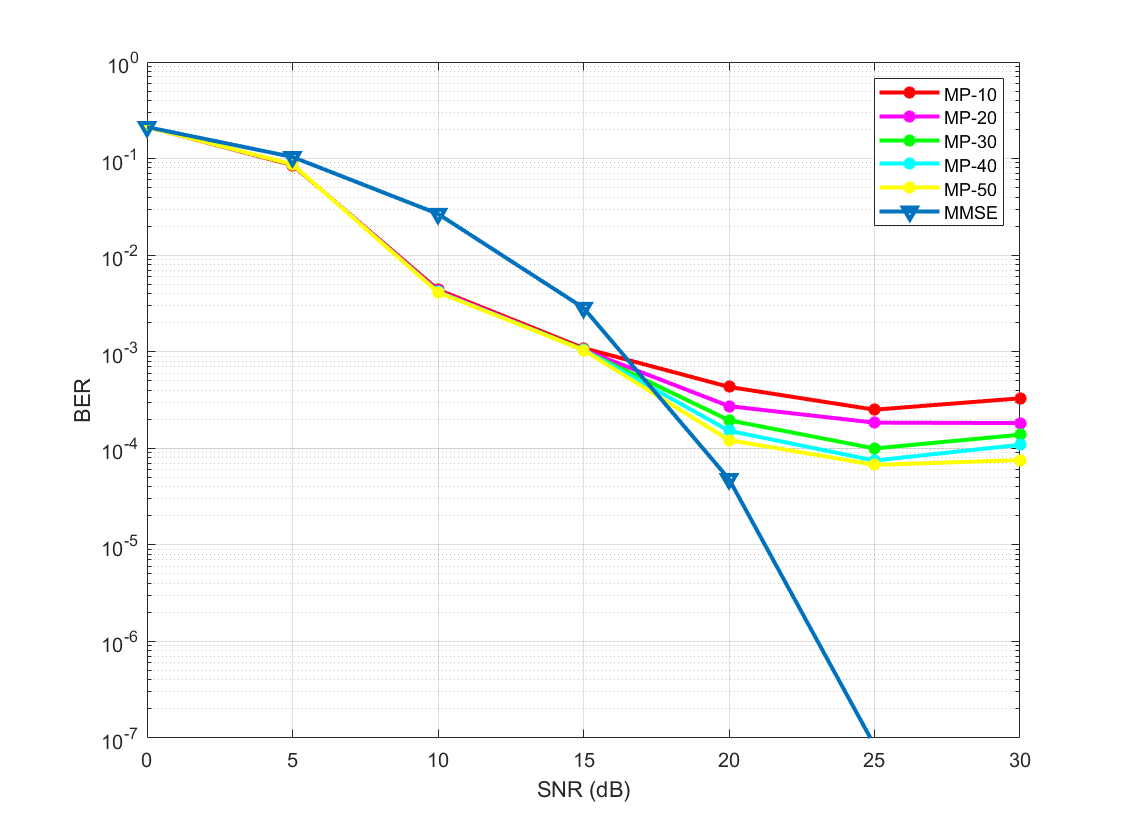}
	\caption{Performance comparison under MMSE and MP equalization.}
	\label{fig:activeIndex}
\end{figure}

\section{Conclusions}
In this paper, we first derived more efficient signal models for fast fading channels, and then proposed new low-complexity and efficient MMSE equalization for both existing OFDM and SC-FDE systems and emerging OTFS systems. Through formulating the frequency-Doppler channel matrix as a circular stripe diagonal matrix, we demonstrated that low-complexity MMSE equalization becomes feasible for systems operating in fast fading channels. We also derived the signal models and proposed MMSE equalization for channels where the Doppler spread is partially-resolvable, which enables effective application of conventional modulations with short signal frames in fast fading channels. In addition, we analyzed the theoretical equalization performance via channel matrix eigenvalue decomposition, providing a useful tool for characterizing the influence of channel conditions and signal modulations on the output SNR after equalization. The simulated BER performance validates that the proposed MMSE equalization can effectively exploit the time diversity with low complexity in fast fading channels for OTFS, OFDM and SC-FDE. The ability of achieving full time and frequency diversity makes OTFS outperform others. We demonstrated that the performance of SC-FDE approaches that of OTFS since it can exploit full frequency diversity and partial time diversity, and OFDM performs the worst due to the lack of frequency diversity though partial time diversity can be achieved with the proposed MMSE equalization. We also simulated the performance with imperfect channel estimation and the results prove the feasibility of MMSE equalization in practice.







%

\appendices

\section{Complexity of Frequency Domain MMSE Equalization}
\par Let us consider the number of complex multiplications and/or divisions associated with calculating the inverse of a $MN\times MN$ circular stripe diagonal matrix with stripe width $1+2k$. Using the Gaussian elimination method, the inversion can be carried out as follows.\\
\begin{figure}[htbp]
	\centering
	\includegraphics[width=1\linewidth]{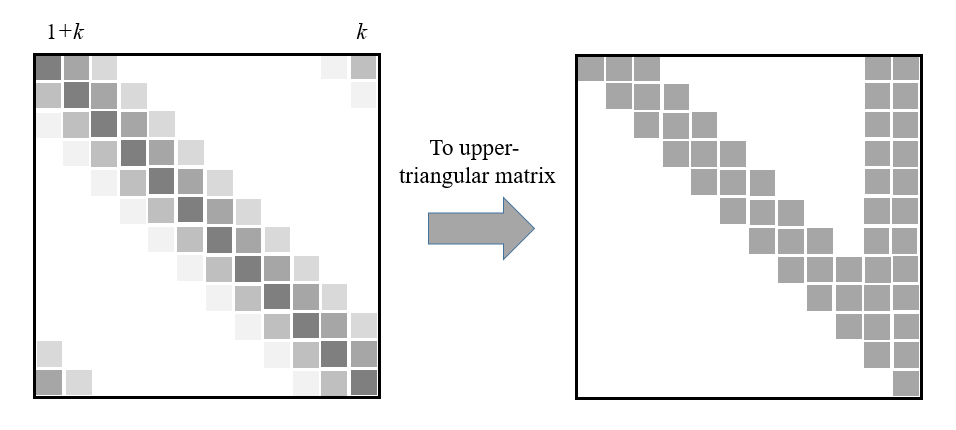}
	\caption{Illustration of Gaussian elimination.}
	\label{Guassian}
\end{figure}
\par (1) We start with the initial footprint of the circular stripe diagonal matrix as shown on the left side of Fig. \ref{Guassian}. For the first row, we  divide all the elements in the row by the first element and cancel the first elements of the other $2k$ rows  with non-zero first elements respectively, which has a complexity of $(1+2k)$ (divisions) $+$ $2k(1+2k)$ (multiplications). This process is repeated from the first to the $(MN-2k)$-th rows. 
\par (2) For the last $2k$ rows, the complexity for canceling all non-zero elements can be expressed as $(2k)^{2}+(2k-1)^{2}+...+1$ to obtain an  upper-triangular matrix with footprint as shown on the right side of Fig. \ref{Guassian}.
\par (3) We then continue to perform the back substitution to turn the upper-triangular matrix to a diagonal matrix. For the first $k$ columns, the complexity can be expressed as $(MN-1)+(MN-2)+...+(MN-k)$.
\par (4) For the $(k+1)$-th to $(MN-k)$-th columns, the complexity can be expressed as $k(MN-2k)$.
\par (5) For the last $k$ columns, the complexity can be expressed as $(k-1)+(k-2)+...+1$.
\par The total computational complexity of matrix inversion is 
\begin{align}\label{eq1}	
C_{total} &= (1+2k)^{2}(MN-2k)+\sum_{n=1}^{2k}n^{2}+\sum_{n=1}^{k}(MN-n)\nonumber\\
&+k(MN-2k)+\sum_{n=1}^{2k}n^{2}\nonumber\\
&=(1+2k)^{2}(MN-2k)+\frac{1}{6}2k(2k+1)(4k+1)\nonumber\\
&+2kMN-2k^{2}-k, 
\end{align}
which has the complexity of $O((1+2k)^{2}MN)$ approximately.  From \eqref{MMSEFDGmatrix}, the matrix $\mathbf{H}_{\nu}$ has a stripe width of $1+2\times K_{max}$, and hence the matrix $(\mathbf{H}_{\nu}\mathbf{H}_{\nu}^{\mathbf{H}}+(1/ \gamma_{in}) \mathbf{I})$ has  a stripe width of $1+4\times K_{max}$. Therefore, the complexity of the frequency domain MMSE equalization is $O((1+4K_{max})^{2}MN)$.

\section{Equivalence of Output SNR Expressions}
\par We first simplify the covariance matrix of $\mathbf{y}$ in \eqref{covariancey} as 
\begin{align}\label{eq1}	
&\mathbf{E\{yy^{H}\}}=\mathbf{C}\sigma_{x}^{2}+\mathbf{D}\sigma_{w}^{2}\nonumber\\
&=\sigma_{x}^{2}(\mathbf{C}+\frac{1}{\gamma_{in}}\mathbf{D})\nonumber\\
&=\sigma_{x}^{2}(\mathbf{V^{H}G}_{t}\mathbf{H}_{t}\mathbf{V}\mathbf{(V^{H}G}_{t}\mathbf{H}_{t}\mathbf{V)^{H}}+\frac{1}{\gamma_{in}}\mathbf{V^{H}G}_{t}\mathbf{(V^{H}G}_{t})^{\mathbf{H}})\nonumber\\
&=\sigma_{x}^{2}\mathbf{V^{H}G}_{t}(\mathbf{H}_{t}\mathbf{H}_{t}^{\mathbf{H}}+\frac{1}{\gamma_{in}}\mathbf{I})\mathbf{G}_{t}^{\mathbf{H}}\mathbf{V}\nonumber\\
&=\sigma_{x}^{2}\mathbf{V^{H}G}_{t}\mathbf{G}_{t}^{-1}\mathbf{H}_{t}^{\mathbf{H}}\mathbf{G}_{t}^{\mathbf{H}}\mathbf{V}\nonumber\\
&=\sigma_{x}^{2}\mathbf{V^{H}}\mathbf{H}_{t}^{\mathbf{H}}\mathbf{G}_{t}^{\mathbf{H}}\mathbf{V}\nonumber\\
&=\sigma_{x}^{2}\mathbf{A^{H}}.
\end{align}
Then, the average total power of the $(nM+m)$-th element in $\mathbf{y}$ can also be expressed as
\begin{equation}\label{eq1}	
\begin{split}
q_{1}[m,n]=\sigma_{x}^{2}\mathbf{A^{H}}[nM+m,nM+m],
\end{split}
\end{equation}
and the output SNR after equalization can be expressed as
\begin{align}\label{eq1}
&\gamma_{out}[m,n]\nonumber\\
&=\frac{|\mathbf{A}[nM+m,nM+m]|^{2}\sigma_{x}^{2}}{\sigma_{x}^{2}\mathbf{A^{H}}[nM+m,nM+m]-|\mathbf{A}[nM+m,nM+m]|^{2}\sigma_{x}^{2}}\nonumber\\
&=\frac{\mathbf{A}[nM+m,nM+m]}{1-\mathbf{A}[nM+m,nM+m]}\nonumber\\
&=\frac{1-J_{nM+m}}{J_{nM+m}}\nonumber\\
&=\frac{1}{J_{nM+m}}-1.
\end{align}

\section*{Acknowledgment}
This research is supported in part by NBN Co under research project PRO18-6384.

\ifCLASSOPTIONcaptionsoff
  \newpage
\fi



%

\bibliographystyle{IEEEtran}
\bibliography{1ref}

\begin{thebibliography}{10}
\providecommand{\url}[1]{#1}
\csname url@samestyle\endcsname
\providecommand{\newblock}{\relax}
\providecommand{\bibinfo}[2]{#2}
\providecommand{\BIBentrySTDinterwordspacing}{\spaceskip=0pt\relax}
\providecommand{\BIBentryALTinterwordstretchfactor}{4}
\providecommand{\BIBentryALTinterwordspacing}{\spaceskip=\fontdimen2\font plus
\BIBentryALTinterwordstretchfactor\fontdimen3\font minus
  \fontdimen4\font\relax}
\providecommand{\BIBforeignlanguage}[2]{{%
\expandafter\ifx\csname l@#1\endcsname\relax
\typeout{** WARNING: IEEEtran.bst: No hyphenation pattern has been}%
\typeout{** loaded for the language `#1'. Using the pattern for}%
\typeout{** the default language instead.}%
\else
\language=\csname l@#1\endcsname
\fi
#2}}
\providecommand{\BIBdecl}{\relax}
\BIBdecl

\bibitem{8515088}
V.~{Khammammetti} and S.~K. {Mohammed}, ``{OTFS}-based multiple-access in high
  doppler and delay spread wireless channels,'' \emph{IEEE Wireless
  Communications Letters}, vol.~8, no.~2, pp. 528--531, April 2019.

\bibitem{5282370}
W.~{Lee} and D.~. {Cho}, ``Mean velocity estimation of mobile stations by
  spatial correlation of channels in cellular systems,'' \emph{IEEE
  Communications Letters}, vol.~13, no.~9, pp. 670--672, Sep. 2009.

\bibitem{7055904}
S.~{Daoud} and A.~{Ghrayeb}, ``Using resampling to combat doppler scaling in
  uwa channels with single-carrier modulation and frequency-domain
  equalization,'' \emph{IEEE Transactions on Vehicular Technology}, vol.~65,
  no.~3, pp. 1261--1270, March 2016.

\bibitem{5524045}
H.~{Nguyen-Le} and T.~{Le-Ngoc}, ``Pilot-aided joint {CFO} and doubly-selective
  channel estimation for {OFDM} transmissions,'' \emph{IEEE Transactions on
  Broadcasting}, vol.~56, no.~4, pp. 514--522, Dec 2010.

\bibitem{1561198}
S.~{Ahmed}, M.~{Sellathurai}, S.~{Lambotharan}, and J.~A. {Chambers},
  ``Low-complexity iterative method of equalization for single carrier with
  cyclic prefix in doubly selective channels,'' \emph{IEEE Signal Processing
  Letters}, vol.~13, no.~1, pp. 5--8, Jan 2006.

\bibitem{5961652}
H.~{Nguyen-Le}, T.~{Le-Ngoc}, and N.~H. {Tran}, ``Iterative receiver design
  with joint doubly selective channel and {CFO} estimation for coded
  {MIMO-OFDM} transmissions,'' \emph{IEEE Transactions on Vehicular
  Technology}, vol.~60, no.~8, pp. 4052--4057, Oct 2011.

\bibitem{7582545}
X.~{Gao}, L.~{Dai}, Y.~{Zhang}, T.~{Xie}, X.~{Dai}, and Z.~{Wang}, ``Fast
  channel tracking for terahertz beamspace massive {MIMO} systems,'' \emph{IEEE
  Transactions on Vehicular Technology}, vol.~66, no.~7, pp. 5689--5696, July
  2017.

\bibitem{4527200}
S.~{Stefanatos} and A.~K. {Katsaggelos}, ``Joint data detection and channel
  tracking for {OFDM} systems with phase noise,'' \emph{IEEE Transactions on
  Signal Processing}, vol.~56, no.~9, pp. 4230--4243, Sep. 2008.

\bibitem{5946232}
W.~{Wang} and S.~S. {Abeysekera}, ``Data aided phase tracking and symbol
  detection for {CPM} in frequency-flat fading channel,'' in \emph{2011 IEEE
  International Conference on Acoustics, Speech and Signal Processing
  (ICASSP)}, May 2011, pp. 3500--3503.

\bibitem{7925924}
R.~{Hadani}, S.~{Rakib}, M.~{Tsatsanis}, A.~{Monk}, A.~J. {Goldsmith}, A.~F.
  {Molisch}, and R.~{Calderbank}, ``Orthogonal time frequency space
  modulation,'' in \emph{2017 IEEE Wireless Communications and Networking
  Conference (WCNC)}, March 2017, pp. 1--6.

\bibitem{8058662}
R.~{Hadani}, S.~{Rakib}, A.~F. {Molisch}, C.~{Ibars}, A.~{Monk},
  M.~{Tsatsanis}, J.~{Delfeld}, A.~{Goldsmith}, and R.~{Calderbank},
  ``Orthogonal time frequency space (otfs) modulation for millimeter-wave
  communications systems,'' in \emph{2017 IEEE MTT-S International Microwave
  Symposium (IMS)}, June 2017, pp. 681--683.

\bibitem{8599041}
P.~{Raviteja}, E.~{Viterbo}, and Y.~{Hong}, ``{OTFS} performance on static
  multipath channels,'' \emph{IEEE Wireless Communications Letters}, vol.~8,
  no.~3, pp. 745--748, June 2019.

\bibitem{8580850}
F.~{Wiffen}, L.~{Sayer}, M.~Z. {Bocus}, A.~{Doufexi}, and A.~{Nix},
  ``Comparison of {OTFS and OFDM} in ray launched sub-6 {GHz and mmWave}
  line-of-sight mobility channels,'' in \emph{2018 IEEE 29th Annual
  International Symposium on Personal, Indoor and Mobile Radio Communications
  (PIMRC)}, Sep. 2018, pp. 73--79.

\bibitem{8503182}
K.~R. {Murali} and A.~{Chockalingam}, ``On {OTFS} modulation for high-{Doppler}
  fading channels,'' in \emph{2018 Information Theory and Applications Workshop
  (ITA)}, San Diego, CA, USA, Feb 2018.

\bibitem{8424569}
P.~{Raviteja}, K.~T. {Phan}, Y.~{Hong}, and E.~{Viterbo}, ``Interference
  cancellation and iterative detection for orthogonal time frequency space
  modulation,'' \emph{IEEE Transactions on Wireless Communications}, vol.~17,
  no.~10, pp. 6501--6515, Oct 2018.

\bibitem{8377159}
P.~{Raviteja}, K.~T. {Phan}, Q.~{Jin}, Y.~{Hong}, and E.~{Viterbo},
  ``Low-complexity iterative detection for orthogonal time frequency space
  modulation,'' in \emph{2018 IEEE Wireless Communications and Networking
  Conference (WCNC)}, Barcelona, Spain, April 2018.

\bibitem{8516353}
P.~{Raviteja}, Y.~{Hong}, E.~{Viterbo}, and E.~{Biglieri}, ``Practical
  pulse-shaping waveforms for reduced-cyclic-prefix {OTFS},'' \emph{IEEE
  Transactions on Vehicular Technology}, vol.~68, no.~1, pp. 957--961, Jan
  2019.

\bibitem{8859227}
S.~{Tiwari}, S.~S. {Das}, and V.~{Rangamgari}, ``Low complexity {LMMSE}
  receiver for {OTFS},'' \emph{IEEE Communications Letters}, 2019, {E}arly
  Access.

\bibitem{8686339}
G.~D. {Surabhi}, R.~M. {Augustine}, and A.~{Chockalingam}, ``On the diversity
  of uncoded otfs modulation in doubly-dispersive channels,'' \emph{IEEE
  Transactions on Wireless Communications}, vol.~18, no.~6, pp. 3049--3063,
  June 2019.

\bibitem{8756831}
E.~{Biglieri}, P.~{Raviteja}, and Y.~{Hong}, ``Error performance of orthogonal
  time frequency space ({OTFS}) modulation,'' in \emph{2019 IEEE International
  Conference on Communications Workshops (ICC Workshops)}, Shanghai, China, May
  2019.

\bibitem{BOOK2}
F.~Hlawatsch and G.~Matz, \emph{Wireless Communications Over Rapidly
  Time-Varying Channels}, 1st~ed.\hskip 1em plus 0.5em minus 0.4em\relax
  Orlando, FL, USA: Academic Press, Inc., 2011.

\bibitem{4392126}
{Xiaojing Huang}, ``{D}iversity performance of precoded {OFDM} with {MMSE}
  equalization,'' in \emph{2007 International Symposium on Communications and
  Information Technologies}, Oct 2007, pp. 802--807.

\bibitem{BOOK1}
J.~G. Proakis, \emph{{D}igital {C}ommunications}, 3rd~ed.\hskip 1em plus 0.5em
  minus 0.4em\relax McGraw-Hill, 1995.

\bibitem{ETSITDL}
ETSI, ``Study on channel model for frequencies from 0.5 to 100 {GHz},''
  \emph{ETSI TR 138 901 V15.0.0}, July 2018.

\bibitem{8855898}
H.~{Zhang}, X.~{Huang}, and J.~A. {Zhang}, ``Comparison of {OTFS} diversity
  performance over slow and fast fading channels,'' in \emph{2019 IEEE/CIC
  International Conference on Communications in China (ICCC)}, Aug 2019, pp.
  828--833.

\end{thebibliography}

%








\end{document}